\documentclass{article}

\usepackage{PRIMEarxiv}

\usepackage[utf8]{inputenc} 
\usepackage[T1]{fontenc}    

\usepackage{mathptmx}

\usepackage[%
	colorlinks=true,%
    	linkcolor=blue,
	citecolor=blue,
    	filecolor=blue,      
    	urlcolor=blue,
    	pdfpagemode=FullScreen,%
]{hyperref}       

\usepackage{url}            

\usepackage{booktabs}       
\usepackage{amsfonts}       
\usepackage{nicefrac}       
\usepackage{microtype}      
\usepackage{fancyhdr}       
\usepackage{graphicx}       

\usepackage{amsthm}
\usepackage{thmtools}

\declaretheorem{lemma}
\declaretheorem{proposition}

\usepackage{amsmath}
\usepackage{natbib}

\usepackage{subcaption}

\graphicspath{{../../}}

\usepackage[linesnumbered, 
boxruled, 
vlined, 
]{algorithm2e}
\DontPrintSemicolon

\SetCommentSty{mycommfont}
\newcommand{\algcom}[1]{\tcp*[r]{#1}}

\usepackage[normalem]{ulem} 
\usepackage[commentmarkup=margin]{changes}
\definechangesauthor[name={Guido}, color=red]{GS}
\definechangesauthor[name={Willem}, color=blue]{WF}

\setauthormarkup{\textsuperscript{\tiny\color{authorcolor}[#1]}}
\setcommentmarkup{\textsuperscript{\tiny\color{authorcolor}[#2\arabic{authorcommentcount}]}\marginpar{\tiny\color{authorcolor}{[#2\arabic{authorcommentcount}]: #1}}}

\newcommand{\cmdfont}[1]{\textsc{#1}}
\newcommand{\pq}{\cmdfont{PQ}}
\newcommand{\ins}{\cmdfont{Insert}}
\newcommand{\del}{\cmdfont{Remove-Min}}
\newcommand{\dec}{\cmdfont{Decrease-Prio}}
\newcommand{\ety}{\cmdfont{Empty}}
\newcommand{\rem}{\cmdfont{Remove}}

\newcommand{\rlx}{\cmdfont{Relax}}
\newcommand{\rlxprediction}{\cmdfont{Relax-Prediction}}
\newcommand{\prediction}{\cmdfont{Prediction}}
\newcommand{\oracle}{\cmdfont{Oracle}}

\newcommand{\Dijkstra}{\cmdfont{Dijkstra}}
\newcommand{\DijkstraPruning}{\cmdfont{Dijkstra-Pruning}}
\newcommand{\DijkstraPrediction}{\cmdfont{Dijkstra-Prediction}}
\newcommand{\DijkstraLearning}{\DijkstraPrediction}
\newcommand{\restart}{\cmdfont{Update-Prediction}}

\newcommand{\len}{\ensuremath{i_{0}}}
\newcommand{\tent}{\ensuremath{tent}}
\newcommand{\Dist}{\ensuremath{D}}
\newcommand{\Pred}{\ensuremath{P}}
\newcommand{\mdelta}{\ensuremath{\varepsilon}}

\newcommand{\event}[1]{\ensuremath{\mathcal{E}_{#1}}}

\newcommand{\inrp}{\text{INRP}}
\newcommand{\inrr}{\text{INRR}}
\newcommand{\inrs}{\text{INRS}}

\newcommand{\pr}[1]{\ensuremath{\mathbf{P}(#1)}}
\newcommand{\prl}[1]{\ensuremath{\mathbf{P}\left(#1\right)}}
\newcommand{\pe}[1]{\ensuremath{\mathbf{E}[#1]}}
\newcommand{\pel}[1]{\ensuremath{\mathbf{E}\left[#1\right]}}
\newcommand{\pcond}{\ensuremath{\; | \;}}

\newcommand{\set}[1]{\{#1\}}
\newcommand{\sset}[2]{\set{#1 \; : \; #2}}

\title{Dijkstra's Algorithm with Predictions to Solve the Single-Source Many-Targets Shortest-Path Problem}

\author{
  Willem Feijen \\
  Networks and Optimization, Centrum Wiskunde \& Informatica (CWI) \\
  Science Park 123, 1098 XG Amsterdam, The Netherlands \\
  \href{mailto:w.feijen@cwi.nl}{\url{w.feijen@cwi.nl}}\\
  \And
  Guido Sch\"afer \\
  Networks and Optimization, Centrum Wiskunde \& Informatica (CWI) \\
  Science Park 123, 1098 XG Amsterdam, The Netherlands \\
  Institute for Logic, Language and Computation (ILLC), University of Amsterdam \\
  Science Park 107, 1098 XG Amsterdam, The Netherlands \\
  \href{mailto:g.schaefer@cwi.nl}{\url{g.schaefer@cwi.nl}} 
}

\begin{document}

\maketitle


\begin{abstract}

We study the use of machine learning techniques to solve a fundamental shortest path problem, known as the \emph{single-source many-targets shortest path problem (SSMTSP)}. Given a directed graph with non-negative edge weights, our goal is to compute a shortest path from a given source node to any of several designated target nodes. 
Basically, our idea is to equip an adapted version of Dijkstra's algorithm with machine learning predictions to solve this problem: Based on the trace of the algorithm, we design a neural network that predicts the shortest path distance after a few iterations. The prediction is then used to prune the search space explored by Dijkstra's algorithm, which may significantly reduce the number of operations on the underlying priority queue. We note that our algorithm works independently of the specific method that is used to arrive at such predictions. Crucially, we require that our algorithm always computes an optimal solution (independently of the accuracy of the prediction) and provides a certificate of optimality. As we show, in the worst-case this might force our algorithm to use the same number of queue operations as Dijkstra's algorithm, even if the prediction is correct.
In general, however, our algorithm may save a significant fraction of the priority queue operations. 
%
%
We derive structural insights that allow us to lower bound these savings on partial random instances. In these instances, an adversary can fix the instance arbitrarily except for the weights of a subset of relevant edges, which are chosen randomly. Our bound shows that the number of relevant edges which are pruned increases as the prediction error decreases. We then use these insights to derive closed-form expressions of the expected number of saved queue operations on random instances. We also present extensive experimental results on random instances showing that the actual savings are oftentimes significantly larger. 

\end{abstract}


\newpage

\setcounter{page}{1}
\section{Introduction}

In recent years, techniques from machine learning (ML) have proven extremely powerful to tackle problems that were considered to be very difficult or even unsolvable. 
Success stories include contributions to health care, natural language processing, image recognition, board games, etc. (see, e.g., \cite{chugh2021survey,lu2007survey,otter2020survey,qayyum2020secure,silver2018general,silver2017mastering}).
As such, machine learning is intimately connected with optimisation because many learning algorithms are based on the optimisation of some loss function over a large set of training samples. Even though optimisation techniques play a vital role in the design of ML-approaches, the reverse direction of using ML-techniques to improve optimisation algorithms is much less explored. 
%

%
In this paper, we contribute to the emerging research agenda of studying how ML-approaches can be used to design improved algorithms for optimisation problems. 
One research direction along these lines has become known as \emph{Algorithms with Predictions} (or \emph{Learning Augmented Algorithms}). This research theme first emerged in the area of online algorithms (see \cite{MV2017,KPS2018,LV2021}), where it has been studied intensively for several years now. 
More recently, similar ideas are investigated also in the context of exact algorithms, data structures, equilibrium analysis and mechanism design. An overview of research articles that appeared on these topics is available at \url{https://algorithms-with-predictions.github.io}. 

In the context of online computation, the idea is to consider online algorithms that can make use of predictions of certain (problem-specific) parameters. For example, such predictions could be obtained from historical input data through the use of machine learning techniques. In fact, these predictions cannot be assumed to be exact as they might be prone to (arbitrarily large) errors. The goal then is to design online algorithms that achieve improved competitive ratios (depending on the quality of the predictions).
Along these lines, the following three objectives emerged in the literature (see also \cite{LV2021}): 
(1) \emph{$\alpha$-consistency}:
If the predictions are perfect (i.e., error-free), the competitive ratio of the online algorithms is bounded by $\alpha$.
(2) \emph{$\beta$-robustness}: 
Even if the predictions are arbitrarily bad, the competitive ratio of the online algorithm is bounded by $\beta$. Ideally, $\beta$ is not (much) worse than the worst-case competitive ratio known for the problem.
(3) \emph{$\gamma$-approximate}: 
If the predictions are within an error of $\eta$ of the actual values (for a suitably defined error parameter $\eta)$, the competitive ratio of the online algorithm is bounded by some function $\gamma(\eta)$. 
Even though the above model was originally introduced for the online setting, it naturally extends to other worst-case performance measures. More recently, researchers started to apply this framework to the study of the running time of algorithms (see, e.g., \cite{NEURIPS2021_5616060f,pmlr-v162-chen22v}), the price of anarchy of equilibria (see, e.g., \cite{GKST2022}) and the social welfare efficiency of truthful mechanisms (see, e.g., \cite{XuLu2022}).

We investigate the use of predictions in the context of a fundamental shortest path problem, which is known as the \emph{single-source many-targets shortest path problem (SSMTSP)}. Given a directed graph $G = (V, E)$ with non-negative edge weights $w: E \rightarrow \mathbb{R}_{\ge 0}$, a source node $s \in V$ and a subset $T \subseteq V$ of designated target nodes, the goal is to compute a shortest path from $s$ to any of the target nodes in $T$. 
Note that the SSMTSP problem generalises the single-source single-target shortest path problem (for which $T = \set{t}$ for a given target node $t \in V$) and the single-source all-targets shortest path problem (for which $T = V$).\footnote{The attentive reader will have noted that the SSMTSP problem can be reduced to a single-source single-target shortest path problem simply by adding a new target node $t$ and connecting each node in $T$ to $t$ with a zero-weight edge. Unless stated otherwise, we adopt the viewpoint of having a set of target nodes throughout the paper.}
These shortest path problems are among the most fundamental optimisation problems with various applications in practice. Further, both the assignment problem and the maximum weight matching problem in bipartite graphs can be reduced to the solving of $n$ SSMTSP problems, where $n$ is the maximum number of nodes on each side of the bipartition (see, e.g., \cite{bast2003heuristic} for details). Clearly, this emphasises the importance of deriving efficient algorithms for this problem. 

The SSMTSP problem can be solved exactly in (strongly) polynomial time by using a slight adaption of the well-known shortest path algorithm due to \cite{dijkstra1959note} (subsequently referred to as \Dijkstra). Basically, \Dijkstra\ grows a shortest path tree rooted at the source node $s$ by iteratively adding new nodes to the tree by increasing distances, until the first target node from $T$ is included.
\Dijkstra\ guarantees a worst-case running time of $O(m + n \log n)$, where $n$ and $m$ denote the number of nodes and edges, respectively. In order to achieve this running time, crucially the underlying priority queue data structure must be implemented through \emph{Fibonacci heaps} (see \cite{fredman1987fibonacci}). In terms of worst-case running time, \Dijkstra\ is the best known algorithm that runs in strongly polynomial time for the shortest path problem with \emph{arbitrary} non-negative edge weights (see, e.g., \cite{cormen2009introduction}). If the edge weights are known to be non-negative integers from a restricted range, there are better algorithms (see also the related work section).

Given that the SSMTSP problem can be solved very efficiently by Dijkstra's algorithm, it is unclear how predictions can help to further improve the running time of the algorithm.\footnote{Note that reading the input instance alone takes time $\Theta(m + n)$. Thus, the best we can hope for is to reduce the $O(n \log(n))$ term in the running time of \Dijkstra\ to $O(n)$.}
In this paper, we focus on exact algorithms for the SSMTSP problem. More specifically, we require that the algorithm computes a shortest path together with a certificate that proves optimality, even if the predictions are arbitrarily bad. A common approach to exhibit such a certificate is by means of a corresponding dual solution to the linear programming formulation of the problem. 
The single-source single-target shortest path problem has a natural (flow-based) linear programming formulation (see, e.g., \cite{PS1998}). The dual of this linear program associates a dual variable $\pi(u)$ with every node $u \in V$ and reads as follows:\footnote{Here, for the sake of conciseness, we state the dual linear program adopting the (equivalent) viewpoint of having a single target node.}
\begin{equation}\label{eq:dual}\tag{D}
\max \quad \pi(t)-\pi(s) \quad   \text{subject to} \quad  \pi(v) \le \pi(u) + w(u,v) \ \ \  \forall (u,v) \in E.
\end{equation}
In this formulation, we can fix $\pi(s) = 0$ without loss of generality. The respective dual values then have a natural interpretation as shortest path distances. By applying the complementary slackness condition, a dual feasible solution $\pi=(\pi_u)_{u \in V}$ proves the optimality of an $s,t$-path if and only if every edge $e = (u,v)$ on that path is \emph{tight}, i.e., $\pi(v) = \pi(u) + w(u,v)$. Most exact shortest path algorithms known in the literature provide a certificate of optimality by constructing such an optimal dual solution (or, equivalently, shortest path distances).


\cite{pmlr-v162-chen22v} recently studied algorithms with predictions for the single-source shortest path problem with arbitrary edge weights under the assumption that a (possibly infeasible) dual solution $\hat{\pi} = (\hat{\pi}(u))_{u \in V}$ is available as prediction. Note that this predictions consists of $n = |V|$ dual values. Building on earlier work by \cite{NEURIPS2021_5616060f}, they show that such predictions can be used to obtain an algorithm to solve this problem in time $O(m \min \set{\| \hat{\pi}-\pi^*\|_1 \cdot \| \hat{\pi} - \pi^*\|_\infty, \sqrt{n}\log(\| \hat{\pi} - \pi^*\|_{\infty})})$, where $\pi^*$ is an optimal dual solution to the problem. In particular, the worst-case running time improves as the difference (evaluated with respect to the $L_1$- and $L_\infty$-norm) between the predicted duals $\hat{\pi}$ and the optimal duals $\pi^*$ decreases. Clearly, such results are appealing as they provide a fine-grained running time guarantee depending on the (cumulative) additive error of the predictions. On the negative side, however, assuming that one has access to the entire dual solution $\hat{\pi}$ might be a rather strong assumption in certain settings---especially because the number of nodes (and thus the number of values to predict) can be very large in practice. 

In this paper, we therefore consider the other end of the spectrum: We assume that our algorithm has access to a \emph{single} predicted value only, namely the shortest path distance of a target node. Said differently, we assume that we have access to a prediction of the objective function value $\hat{\pi}_t$ of \eqref{eq:dual} (which might not necessarily coincide with the optimal objective function value $\pi_t^*$). The question that we address here is whether such a minimalistic prediction suffices to still achieve a running time improvement of Dijkstra's algorithm. 

Clearly, exploiting a prediction of the shortest path distance only is much more restrictive than assuming that one has access to an entire dual solution---in fact, the latter can be used to derive the former. Intuitively, it is clear that it is more challenging to improve Dijkstra's algorithm using such an (inferior) prediction. The following observation lends further support to this intuition. Recall that we require that our algorithm provides a certificate of optimality. 
If the entire dual solution is predicted perfectly, this becomes trivial: the dual itself provides such a certificate and is readily available in this case. 
On the other hand, this remains a challenging task in our setting: even if the correct shortest path distance is known, it is non-trivial (in terms of computational work) to compute a corresponding dual solution that proves optimality. 
This latter observation will be made more rigorous in Section~\ref{s:lower-bound}. 

%
%

Our algorithm is based on a heuristic improvement of Dijkstra's algorithm proposed by \cite{bast2003heuristic} (referred to as \DijkstraPruning). 
The key idea of their algorithm is to exploit that Dijkstra's algorithm might encounter many target nodes in $T$ before the actual target node (defining the shortest path distance) is added to the tree: \DijkstraPruning\ simply keeps track of the smallest distance $B$ to a node in $T$ that was encountered so far. The algorithm then \emph{prunes} each edge that leads to a node whose distance exceeds $B$---clearly, these edges are irrelevant for the final shortest path. 
A more detailed description of this algorithm will be given in Section~\ref{sec:preliminaries}.
Although the pruning idea is very simple, it can significantly improve the running time of the algorithm.  
In the worst case, however, \DijkstraPruning\ has the same running time as \Dijkstra. Therefore, \cite{bast2003heuristic} investigate the effectiveness of \DijkstraPruning\ on random instances, both analytically and experimentally. They show that the pruning of irrelevant edges significantly reduces the number of executed priority queue operations on these instances.

\subsection{Our Contributions}
The main contributions presented in this paper are as follows: 
\begin{enumerate}
\item 
We combine an ML-approach with the pruning idea above to obtain a new algorithm for the SSMTSP problem. 
Basically, our algorithm (referred to as \DijkstraPrediction) computes a prediction $\Pred$ of the final shortest path distance after a few iterations, and then uses this prediction $\Pred$ together with the pruning trick of \cite{bast2003heuristic} to further reduce the search space it explores. 
On the ML-side, one of the challenges is to define features capturing the essence of the Dijkstra run which can be used to arrive at a good prediction. 
On the algorithm-design side, we need to tackle the problem that the prediction might be an underestimation of the actual shortest path distance. 
We note that our algorithm works independently of the specific method that is used to arrive at the prediction $P$. 

\item 
We prove that our new algorithm \DijkstraPrediction\ always computes an exact solution and has a worst-case running time of $O(m + n \log n)$ (independent of the prediction error). 
In particular, \DijkstraPrediction\ retains the best worst-case running time and always provides a certificate of optimality. That is, while our algorithm will never use more priority queue operations than the adapted Dijkstra algorithm, it can potentially save many queue operations additionally.

\item 
We establish a lower bound (both in expectation and with high probability) on the number of edges pruned by our new algorithm \DijkstraPrediction. 
Our bound depends on the number of \emph{relevant} edges (leading to nodes whose distances exceed the shortest path distance) and increases as the prediction error decreases.
Our bound applies to arbitrary instances as long the weights of the relevant edges are chosen at random---we refer to this setting as the \emph{partial random model}.
Further, we show that no improvement is possible in the worst case, even if the prediction is perfect---this also justifies the use of our partial random model. 

\item We then derive a bound on the expected number of queue operations saved by \DijkstraPrediction\ in comparison to \DijkstraPruning\ and \Dijkstra\ on random instances. 
While \DijkstraPruning\ already significantly improves over \Dijkstra, we show that \DijkstraPrediction\ further reduces the number of nodes which are inserted but never removed from the queue. 
More specifically, we consider Erd\"os-R\'enyi random graphs with average degree $c$ and uniform edge weights in $[0,1]$, where the source node is chosen uniformly at random and each node is selected as a target node with probability $q$ (formal definitions are given below).
If $\Dist$ denotes the shortest path distance and $\mdelta$ denotes the (additive) error of the prediction with $D < 1 - \mdelta$, we show that the number of nodes inserted but never removed from the priority queue by \DijkstraPrediction, \DijkstraPruning\ and \Dijkstra, respectively, is at most
\[
\frac{1}{q}\left(1 + \ln(c-1) - \ln\left(\frac{1-D}{\mdelta}\right)\right), \quad
\frac{1}{q}(1 + \ln(c-1)) \quad \text{and} \quad 
\frac{c-1}{q}.
\]
Here the latter two bounds were established by \cite{bast2003heuristic}.
%
Technically, this is the most challenging part of our analysis as we need to estimate the savings incurred by the pruning bound $B$ as well as the prediction bound $\Pred$ in one probabilistic argument.

\item 
We also report on our extensive experimental studies on random graphs to complement our theoretical findings above. We compare our new algorithm \DijkstraPrediction\ to the existing ones and evaluate different prediction algorithms. 
Our experiments show that \DijkstraPrediction, which combines 
a neural network ML-approach to compute the prediction with an \restart\ procedure to handle possible underestimations, significantly outperforms all other algorithms (i.e., combinations of different prediction algorithms). 
\end{enumerate}

\subsection{Related Work}

\subparagraph*{Algorithms with predictions.}
Using ML techniques in combinatorial algorithms has been studied intensively recently. 
We refer the reader to \cite{bengio2020machine} for a survey paper on leveraging ML to solve combinatorial optimisation problems. 
In this survey, three different approaches of using ML components in combinatorial optimisation algorithms are given. Our approach falls into the second of the three approaches, in which meaningful properties of the optimisation problem are learnt and used to augment the algorithm. 

The line of research known as \emph{Algorithms with Predictions} falls into this second approach and aims to achieve near optimal algorithms when the predictions are good, while falling back to the worst-case behaviour if the prediction error is large (see, e.g., \cite{algorithmspredictions, onlinescheduling} and the references provided above).
This idea is applied to optimisation problems like the \emph{ski rental problem}, \emph{caching problem} and \emph{bipartite matching}. 

\cite{NEURIPS2021_5616060f} use predictions to improve the worst case running time to solve the bipartite matching problem. Their main idea is to use predictions for the dual values as a warm start of the primal-dual algorithm. Since the predicted duals might not be feasible, they propose a rounding procedure to compute a feasible dual, close to the predicted one. Moreover, they prove that the prediction of the duals that they require for the algorithm can actually be learned, by showing that this prediction problem has low sample complexity.

\cite{pmlr-v162-chen22v} builds further upon the results of \cite{NEURIPS2021_5616060f}. First, they give an improvement of the algorithm for bipartite matching, which reduces the worst case running time even more. Secondly, they extend the idea of using predictions for primal-dual algorithms and apply it to a shortest path problem. When the predictions are accurate enough, they achieve an almost linear running time. Further, they propose a general reduction-based framework for learning-based algorithms and extend the PAC-learnability results of \cite{NEURIPS2021_5616060f} beyond the bipartite matching problem.

Our paper differs from \cite{pmlr-v162-chen22v} since our algorithm only requires a single prediction, for the shortest path value, instead of a learned dual for each node. 
That is, our algorithm requires less in terms of prediction, on the other hand, our algorithm does not improve the worst case running time. Instead, we can prove a lower bound on the number of expensive priority queue operations that are saved. 

\subparagraph*{Classical shortest path algorithms.}
An extensive survey of combinatorial algorithms to solve the shortest path problem is given by \cite{madkour2017survey}. We give a short summary of their extensive report, touching upon different used shortest path techniques by grouping the methods in four categories. 

As explained above, \cite{fredman1987fibonacci} improve Dijkstra by introducing Fibonacci heaps. Alternative heap structures that gave further improvements are AF-heaps \cite{fredman1990blasting, fredman1990trans, fredman1993surpassing} and relaxed fibonacci heaps \cite{driscoll1988relaxed}. An implementation based on stratified binary trees is introduced by \cite{van1975preserving}. \cite{thorup1999undirected} indicates there is an analogy between sorting and Single Source Shortest path, claiming that SSSP is no harder than sorting the edge weights. \cite{han2001improved} improves on these results. \cite{thorup1999undirected} builds hierarchical bucketing structure, which is improved by \cite{hagerup2000improved}. 

The second category contains the \emph{Distance Oracle} algorithms, introduced by \cite{thorup2005approximate}, which consist of a pre-processing phase and a query phase. In the pre-processing phase an auxilliary data structure is constructed, which is queried in the query phase to compute the shortest path. Distance oracle algorithms can be both exact, like in \cite{fakcharoenphol2006planar} or approximate, like in \cite{elkin20041}. Some methods approximate distance using a \emph{spanner}, a subgraph that maintains the locality aspects of the original graph. Other methods approximate distances using a \emph{landmark approach}, where each vertex stores distances to a set of chosen landmarks \cite{sommer2014shortest}.
All distance oracle algorithms deal with a trade-off between space complexity and query time.

\emph{Goal-Directed} Shortest Path algorithms fall in the third category. These algorihtms add annotations to vertices or edges with additional information. This allows the algorithm to determine which part of the graph to search in the search phase, and which parts to prune. A well known algorithm in this category is A*, which, unlike Dijkstra, is an informed algorithm, since it searches the route which leads to the goal. If an admissible heuristic is used, A* will return the optimal shortest path, but it might fail if the heuristic does not work well. Several variants and improvements to A* have been proposed, which include landmark approaches \cite{goldberg2005computing} or the concept of \emph{reach} \cite{gutman2004reach}. Intuitively, the reach of a vertex encodes the lengths of shortest paths on which this vertex lies. Other goal-directed methods include edge labels \cite{kohler2005acceleration, schulz2000dijkstra, ich2006extremely}, the arc flag approach \cite{mohring2007partitioning, hilger2009fast, bauer2010sharc} or pre-computed cluster distances \cite{maue2010goal}. In this last method, the graph is partitioned in clusters, after which the shortest connection between clusters is stored. 
Interestingly, also for $A^*$, recent research shows an interest for replacing heuristics with machine learning. In \cite{eden2022embeddings}, estimates in $A^*$ that were formerly done with heuristics are executed with learning techniques, based on features of the nodes. They find there is a trade-off between the amount of information used to describe a node and the improvement in running time of the algorithm. 

The last category in this non-extensive list of shortest path methods are the \emph{Hierarchical} shortest path methods. These methods are prominent for problems which naturally exhibit a hierarchical structure, like road networks. 
An example of a hierarchical method are the highway hierarchies, which label an edge on a shortest path as a highway if it is not in the proximity of the source or target, as done in \cite{sanders2005highway, sanders2006engineering}.
Other hierarchical methods are contraction hierarchies \cite{geisberger2008contraction, geisberger2012exact} and hub labelling \cite {gavoille2004distance, thorup2005approximate}. 



%

\subparagraph*{Approximating shortest paths using ML.}
Next to classical combinatorial approaches, there has also been great interest from the field of ML in finding approximates for the shortest path distance, using an ML perspective. 
For example, \cite{bagheri2008finding} compute shortest paths by using a genetic algorithm. 
Their algorithm works faster than Dijkstra, but they only test on small graphs with at most 80 nodes. Also, more recently, using ML techniques to approximate shortest path distances has lead to interesting results. 
For example, \cite{Rizi_2018} create an estimate for the shortest path distance between two nodes in a two step procedure: first a deep learning vector embedding is applied and then a well-known landmark procedure afterwards 
(see, e.g.,  \cite{zhao2010orion, zhao2011efficient}). 
\cite{Rizi_2018} show results on large-scale real-world social networks with more than one million nodes. Their method  differs from our approach in the sense that an algorithm is created to approximate shortest path distances in one specific large-scale real-world graph, opposed to an algorithm which can be used for any graph from a set of random graphs with similar properties. 



\subsection{Organisation of Paper} 
The paper is organised as follows: 
In Section \ref{sec:preliminaries}, we formally define the problem, describe the adapted Dijkstra algorithm on which our algorithms are based on, and introduce the random graph model that we use in this paper. 
In Section \ref{sec:DijkstraML}, we introduce our new algorithm that combines the edge pruning idea implemented by the adapted Dijkstra algorithm with shortest path predictions and prove its correctness. We also give an \restart\ procedure to handle underestimations of the shortest path distance. 
In Section \ref{s:lower-bound}, we prove a lower bound on the number of saved queue operations if the edge weights are chosen at random. We apply this bound to estimate the savings on sparse random graphs. 
In Section \ref{s:prediction_methods}, we elaborate on different prediction methods (both ML-based and based on breadth-first search); we remark that our algorithm can be used with arbitrary prediction algorithms. 
Finally, in Section \ref{sec:experiments} we describe our experimental setup and report on the respective findings. 

\section{Preliminaries}
\label{sec:preliminaries}


The \emph{single-source many-targets shortest path problem (SSMTSP)} has been defined in the introduction. 
We use $n$ and $m$ to refer to the number of nodes and edges of the underlying graph $G$, respectively. 
For every node $v \in V$, we use $\delta(v)$ to denote the total weight of a shortest path (with respect to $w$) from $s$ to $v$; if $v$ cannot be reached from $s$ we adopt the convention that $\delta(v) = \infty$. Given that all edge weights are non-negative, we thus have $\delta(v) \in \mathbb{R}_{\ge 0} \cup \set{\infty}$.
Note that to solve the SSMTSP problem it is sufficient to compute the shortest path distances of all nodes $v \in V$ satisfying $\delta(v) \le \Dist$, where $\Dist$ is the minimum shortest path distance of a target node, i.e., 
$\Dist = \min_{t \in T} \delta(t).$
Once these distances are computed, the actual shortest path can be extracted in linear time $O(n+ m)$ by computing the \emph{shortest path tree} rooted at $s$ (see, e.g., \cite{cormen2009introduction} for more details). Throughout this paper, we assume that there is at least one target node in $T$ that is reachable from $s$.\footnote{Note that this can easily be checked in linear time $O(n+m)$, simply by running a \emph{breadth-first search (BFS)} (\cite{cormen2009introduction}).}

\subsection{\DijkstraPruning\ Algorithm}
As mentioned in the introduction, our algorithm combines an adaptation of Dijkstra's algorithm by~\cite{bast2003heuristic} (referred to as \DijkstraPruning) with an ML-prediction. We briefly review the adaptation here. 


We first describe the standard Dijkstra algorithm (referred to as \Dijkstra), adapted to many targets. 
\Dijkstra\ associates a \emph{tentative distance} $d(v)$ with every node $v \in V$ and maintains the invariant that $d(v) \ge \delta(v)$ for every $v \in V$. Initially, $d(s) = 0$ and $d(v) = \infty$ for all $v \in V$. The set of nodes is partitioned into the set of \emph{settled} and \emph{unsettled} nodes. Initially, all nodes are unsettled, and whenever the algorithm declares a node $v$ to be settled, its tentative distance is exact, i.e., $d(v) = \delta(v)$. The algorithm maintains a priority queue $\pq$ to keep track of the distance labels of the unsettled nodes $v \in V$ with $d(v) \neq \infty$. Initially, only the source node $s$ is contained in $\pq$. 
In each iteration, the algorithm removes from $\pq$ an unsettled node $u$ of minimum tentative distance, declares it to be settled and scans each outgoing edge $(u,v) \in E$ to check whether $d(v)$ needs to be updated; we also say that edge $(u,v)$ is \emph{relaxed} (pseudocode in Algorithm \ref{alg:relax}). The algorithm terminates when a node $u \in T$ becomes settled. 
In the worst case, \Dijkstra\ performs $n$ \del, $n$ \ins\ and $m$ \dec\ operations. Its running time crucially depends on how efficiently these operations are supported by the underlying priority queue data structure. 
In this context, \emph{Fibonacci heaps} introduced by \cite{fredman1987fibonacci} are the (theoretically) most efficient data structure, supporting all these operations in (amortised) time $O(m + n \log n)$. It is important to realise though that the actual time needed by the queue operations depends on the size (i.e., number of elements) of the priority queue. In general, a smaller queue size results in a better overall running time of the algorithm.

\DijkstraPruning\ works the same way as \Dijkstra, but additionally keeps track of an \emph{upper bound $B$} on the shortest path distance to a node in $T$.
Initially, $B = \infty$ and the algorithm lowers this bound whenever a shorter path to a node in $T$ is encountered. Crucially, $B \ge \Dist$ always, 
and as a consequence, each edge $(u,v) \in E$ that leads to a finite tentative distance $d(v)$ with $d(v) \ge B$ can be discarded from further considerations; we also say that edge $(u,v)$ is \emph{pruned}. 
The pseudocode of \DijkstraPruning\ is given in Algorithm \ref{alg:dijkstra-improved}.
Clearly, in the worst case \DijkstraPruning\ does not prune any edges. 
In particular, the worst-case running time of \DijkstraPruning\ remains $O(m + n \log n)$.

\begin{algorithm}[t]
 	$d(s) = 0$, $d(v) = \infty$ for all $v \in V \setminus \{s\}, B=\infty$ \algcom{tentative distances and pruning bound}
	$\pq.\ins(s, d(s))$ \algcom{priority queue}
	\While{not $\pq.\ety()$}{
		$u=\pq.\del()$ \algcom{$u$ becomes settled}
		\lIf(\tcp*[f]{stop when $u$ is target}){$u \in T$}{\textbf{STOP}}
		\For
		{$(u,v) \in E$}{
                 	$\tent = d(u) + w(u,v)$ \algcom{tentative distance of $v$}
			\lIf(\tcp*[f]{prune edge $(u,v)$}){$\tent > B$}{\textbf{continue}}
			\lIf(\tcp*[f]{update $B$}){$v \in T$}{$B=\min\{\tent, B\}$}
			$\rlx(u,v,\tent)$ \;
		} 
	} 
	\caption{$\DijkstraPruning(G,w,s, T)$}
    	\label{alg:dijkstra-improved}
\end{algorithm}

\begin{algorithm}[t]
	\If(\tcp*[f]{lower tentative distance}){$d(v) > \tent$}{
		\lIf(\tcp*[f]{add $v$ to $\pq$}){$d(v) = \infty$}{$\pq.\ins(v, \tent)$ \label{alg:relax:insert}\label{alg:relax:infinite-d}}
        		\lElse(\tcp*[f]{decrease priority of $v$}){$\pq.\dec(v, \tent)$}
		$d(v) = \tent$ \algcom{update distance of $v$}
	}
    \caption{$\rlx(u,v,\tent)$}
    \label{alg:relax}
\end{algorithm}

\subsection{Random Model}
\cite{bast2003heuristic} use the following random model to analyse the improved performance of \DijkstraPruning\ and show that the expected savings for these instances are significant, both analytically and empirically. 
The directed random graph instances are constructed using the \emph{Erd\"os-R\'enyi random graph model} by \cite{Gilbert1959}, also known as $G(n, p)$: there are $n$ nodes and each of the $n(n-1)$ possible (directed) edges is present independently with probability $p = c/n$, where $c$ is (roughly) the average degree of a node.
Further, each node $u \in V$ is chosen independently with probability $q = f/n$ to belong to the target set $T$, where $f$ is the expected number of target nodes in $T$.
The weight $w(e)$ of each edge $e$ is chosen independently uniformly at random from the range $[0,1]$.

\section{Dijkstra's Algorithm with Predictions}
\label{sec:DijkstraML}


\begin{algorithm}[t]
 	$d(s) = 0$, $d(v) = \infty$ for all $v \in V \setminus \{s\}$\algcom{tentative distances} 
	$B = \infty$, $P = \infty$, $X = []$, $i = 0$ \algcom{pruning bound, prediction and array to store trace} 
         $\pq.\ins(s, d(s))$, $R = \emptyset$  \algcom{$\pq$ and reserve set $R$}
	\While{\textrm{\bf not} $\pq.\ety()$ \textrm{\bf and} $\pq.\textsc{min-prio}() \le \Pred$}{\label{alg:dijkstra-learn:while}
		$u=\pq.\del()$ \algcom{$u$ becomes settled}
		\lIf(\tcp*[f]{stop when $u$ is target}){$u \in T$}{\textbf{STOP} \label{alg:dijkstra-learn:return}}
                 $i = i + 1$\;
                 \lIf(\tcp*[f]{extend trace}){$i\leq \len$}{$X[i] = (d(u), B)$ \label{alg:dijkstra-learn:updateX}}
                 \lIf(\tcp*[f]{get prediction}){$i = \len$}{$P = \alpha \cdot \prediction(X)$ \label{alg:dijkstra-learn:makep}}
                 \For
                 {$(u,v) \in E$}{
                 	$\tent = d(u) + w(u,v)$ \algcom{tentative distance of $v$}
			\lIf(\tcp*[f]{prune $(u,v)$}){$\tent > B$}{\textbf{continue}}\label{alg:dijkstra-learn:boundp}
			\lIf(\tcp*[f]{update $B$}){$v \in T$}{$B=\min\{\tent, B\}$}
			$\rlxprediction^{\star}(u,v,\tent, P)$ \algcom{call relax routine}
		}
	}	 
	\restart \label{alg:dijkstra-learn:restart} \algcom{call update prediction routine} 
	\textbf{continue} with while-loop
    \caption{$\DijkstraPrediction(G, w, s, T, \len, \alpha, \beta)$}
    \label{alg:dijkstra-learn}    
\end{algorithm}

%
Our basic idea is to further amplify the effect of the edge prunings by using a machine learning approach to obtain a prediction of the shortest path distance at an early stage. 
More concretely, suppose we have a \prediction\ algorithm which, based on the execution of the algorithm so far, computes an estimate of the shortest path distance $\Dist$. We can then call this algorithm after a few iterations to obtain a prediction $\Pred$ of $\Dist$ and use it to prune all edges that lead to a tentative distance larger than $\Pred$. 
There are three main advantages from which our approach can (potentially) benefit when compared to the algorithms \Dijkstra\ and \DijkstraPruning. Firstly, fewer queue operations may be performed because of the edges being pruned. Secondly, edge pruning might start after a few iterations only, potentially before having found any path to a target node and finally, queue operations may take less time because the size of the priority queue remains smaller. 


\subsection{Detailed Description of \DijkstraPrediction}

%

We elaborate on our algorithm \DijkstraPrediction\ (Algorithm \ref{alg:dijkstra-learn})  in more detail. 
The algorithm builds upon \DijkstraPruning, see Section~\ref{sec:preliminaries}.
The three new input parameters $\len, \alpha$ and $\beta$ will become clear below. 
During the first $\len$ iterations, an array $X$ is maintained for storing the \emph{trace} (as we term it) of the algorithm. 
In iteration $\len$, the constructed trace $X$ is then used to compute an initial prediction by calling the \textsc{Prediction} procedure, for which several alternatives are given in Section \ref{s:prediction_methods}.
The algorithm keeps track of both the bound $B$ on the smallest distance to a node in $T$ encountered so far and the current prediction $\Pred$. A scanned edge $(u,v)$ is not inserted into the priority queue, we call this \emph{pruning}, whenever its tentative distance $d(u) + w(u,v)$ exceeds $B$ or $P$.%
\footnote{There is a somewhat subtle point in the algorithm: Note that during the first $\len$ iterations the prediction $\Pred$ remains at $\infty$ as the trace is just being built. As a consequence, throughout this stage it could happen that nodes are inserted into the priority queue $\pq$, whose tentative distances are larger than the first prediction $\Pred$ (determined in iteration $\len$). After this stage, this is impossible due to the pruning. 
It is because of these nodes that we have to add the second condition to the while loop, which checks whether the minimum distance of a node in $\pq$ is less than the current prediction $\Pred$. If not, the \restart\ procedure has to be initiated to increase the prediction and add all relevant nodes to $\pq$.}

Ideally, we would like to come up with a \textsc{Prediction} procedure that provides a prediction $\Pred$ which comes close to the actual shortest path distance $\Dist$. In fact, both over- and underestimations of $\Dist$ can be harmful, though in different ways: If $\Pred$ overestimates $\Dist$ then edges which are irrelevant for the shortest path might not be pruned and the algorithm might perform redundant operations---which is undesirable. If $\Pred$ underestimates $\Dist$ then edges which are essential for the shortest path might be pruned and an incorrect solution might be returned---which is unacceptable. 
To remedy the latter, we equip our algorithm with an \restart\ procedure (Algorithm \ref{alg:dijkstra-smart}): If the prediction $\Pred$ turns out to be too small, it is increased by a factor $\beta > 1$ and the algorithm continues.
Clearly, such \restart\ procedures should not happen too often as this might reduce the efficiency of the approach, therefore, the initial prediction is slightly inflated by a factor $\alpha \ge 1$. 
By inflating the prediction in an \restart\ routine, nodes that were previously considered irrelevant, could potentially become relevant. Here a node is considered \emph{relevant} if its tentative distance is smaller than the updated prediction. We need to insert the nodes that have become relevant during the \restart\ routine into the priority queue, we call this a \emph{batch insertion}. During a batch insertion, we only insert a node if its tentative distance does not exceed the current upper bound $B$. 

We are able to efficiently execute a batch insertion by maintaining a set of \emph{reserve nodes}, $R$, during the algorithm. $R$ will contain all nodes which have a finite tentative distance, but have not been added to the priority queue because their tentative distance exceeds the prediction in the current trial. 
Maintaining set $R$ is done by using a different relax routine than \DijkstraPruning, namely \rlxprediction\ and by using a hand-tailored data structure, both on which we elaborate below.




\begin{algorithm}[t]
        \If(\tcp*[f]{lower tentative distance}){$d(v) > \tent$ }
        {
            \eIf(\tcp*[f]{$v$ reached for the first time}){$d(v) = \infty$}
            {
                \lIf(\tcp*[f]{add $v$ to $\pq$}){$\tent \le P$}{$\pq.\ins(v, tent)$}
		        \lElse(\tcp*[f]{add $v$ to $R$}){$R.\ins(v, tent)$}
           }
           {
                \lIf
                {$v\not\in R$}{$\pq.\dec(v, tent)$}
		        \Else(\tcp*[f]{$v$ is in $R$ already}){
                    \If(\tcp*[f]{$v$ remains in $R$}){$\tent > P$}
                    {
                        $R.\dec(v, tent)$
                        \textbf{continue}
                    }
                    $R.\rem(v)$ \algcom{move $v$ from $R$ to $\pq$}
                    $\pq.\ins(v, \tent)$ \;
                }
           }
           $d(v) = \tent$ \algcom{update distance of $v$}
       }
    \caption{$\rlxprediction(u,v,\tent,P)$}
    \label{alg:Relax-bound}
\end{algorithm} 
%
\begin{algorithm}[t]
        $P= \beta \cdot P$ \algcom{inflate prediction}
        \ForEach(\tcp*[f]{iterate over all relevant nodes in $R$ and do batch insertion}){$v\in R$}{ \label{alg:dijkstra-smart:start-batch-insert}
            \If{{$d(v) \leq B$}}{
             $R.\rem(v)$  \algcom{move $v$ from $R$ to $\pq$}
             $\pq.\ins(v, d(v))$              
           }
        } \label{alg:dijkstra-smart:end-batch-insert}
        
    \caption{\restart}
    \label{alg:dijkstra-smart}
\end{algorithm}


The \rlxprediction\ routine is similar to the standard \rlx\ routine (see Algorithm \ref{alg:Relax-bound} and Algorithm \ref{alg:relax}). The main difference is that the node $v$ is only inserted into the priority queue if its tentative distance is smaller than the prediction; otherwise, it is inserted into the reserve set $R$. 

By using a tailored data-structure for the reserve set, we can quickly execute the batch insertions. 
In this tailored data-structure we store nodes based on their tentative distance, like in the priority queue. 
However, unlike in the priority queue, the nodes are not \emph{sorted} based on this tentative distance. 
Instead, nodes are stored in several linked lists, which we call \emph{buckets}.
Each bucket has a bucket number $i\in \mathbb{N}$, and a node will be stored in bucket $i = 1, 2, \dots$ if and only if it has a tentative distance $\tent$ such that $\beta ^{i -1}\Pred < \tent \leq \beta ^ {i} \Pred$. 
Then, during the $i$'th batch insertion, simply all the nodes from bucket $i$ can be moved from $R$ into $\pq$.


\subsection{Correctness Proof}
An important point is that our algorithms is \emph{correct} in the sense that it 
\begin{enumerate}
\item terminates in polynomial time, and
\item computes an optimal solution to the SSMTSP problem. 
\end{enumerate}

The following theorem can be proven by relating the runs of \DijkstraPrediction\ with \restart\ and \DijkstraPruning. 

\begin{restatable}{theorem}{correctness}
\label{thm:correctness}
\DijkstraPrediction\ is correct and has a worst-case running time of $O(m + n \log n)$.
\end{restatable}

The proof of Theorem~\ref{thm:correctness} follows directly from the following invariant, which establishes a connection between \DijkstraPrediction\ with \restart\ and \DijkstraPruning, and from Lemma \ref{lemma-reserveset}, which establishes that the reserve set operations do not increase the worst-case running time.

\begin{restatable}{invariant}{invariantCorrectness}
\label{invariant-all}
\sloppy
Consider the runs of \DijkstraPrediction\ and  \DijkstraPruning\ on the same input instance. We use $d$, $\pq$ and $R$ to refer to the respective data structures in \DijkstraPrediction, and $d'$, $\pq'$ to the respective data structures in \DijkstraPruning.
The following properties are satisfied in each iteration: 
\setlength{\leftmargini}{3em}
\begin{itemize}
    \item[(P1)]  Both algorithms remove the same node $u$ from $\pq$ and $\pq'$, respectively. 
    \item [(P2)] The set of nodes in $\pq'$ can be partitioned into the set of nodes in $\pq$ and the set of nodes in $R$, with $d(v) > \Pred$ for all $v \in R$.
    \item [(P3)] The tentative distances are equal in both algorithms, i.e.,  $d(v)=d'(v)$ for all $v \in V$. 
\end{itemize}
\end{restatable}

\begin{proof}
We assume that both algorithms employ a consistent tie-breaking rule for nodes with similar distances.
It is easy to see that the invariant holds for the first iteration: the prediction is initialised to $\Pred = \infty$ and thus the algorithms do exactly the same, and $R$ remains empty. 
Now, suppose by induction that the invariant holds at the beginning of iteration $i \ge 1$. We argue that the invariant holds at the end of iteration $i$:

(P1) Suppose that node $u$ is deleted from $\pq$ in iteration $i$. By the condition in the while-loop in Algorithm \ref{alg:dijkstra-learn} it holds that $d(u) \leq P$, which together with (P2) gives that $d(u) < d(v)$ for all $v\in R$. By the \del\ operation, $d(u) < d(v)$ for all $v \in \pq$, so $d(u) < d(v)$ for all $v \in \pq \cup R$. Since $\pq\cup R=\pq'$ (because of (P2)), we have $d(u) < d(v)$ for all $v \in \pq'$. From (P3), it follows that $d'(u) < d'(v)$ for all $v \in \pq'$, which proves that the same node is deleted in \DijkstraPruning. 

(P2) We consider each queue operation executed by the algorithms in this iteration separately and argue that the claim remains true. 
Firstly, if the claim holds at the beginning of an iteration, then it still holds after the \del\ operation because the same node $u$ is deleted from $\pq$ and $\pq'$. 
Secondly, suppose that in iteration $i$ node $v$ is inserted into $\pq'$ because edge $(u,v)$ is relaxed. Then, before the insertion, $d'(v)=\infty$ and $d'(u) + w(u,v)<\infty$. Edge $(u,v)$ will also be relaxed in the \DijkstraPrediction\ algorithm. By (P3), we have $d(v)=\infty$ and $d(u) + w(u,v)<\infty$. This means that node $v$ is either inserted into $\pq$ or $R$. A node $v$ is only added to $R$ if $d(u) + w(u,v) > P$.
So the claim still holds after an insertion when $d(v)$ is set to $d(u) + w(u,v)$. 
Thirdly, suppose that in iteration $i$ the tentative distance of node $v$ in $\pq'$ is decreased because edge $(u,v)$ is relaxed. Then, before the tentative distance is decreased, $d'(v)<\infty$ and $d'(u) + w(u,v)<d'(v)$. Again, edge $(u,v)$ will also be relaxed in the \DijkstraPrediction\ algorithm, and by (P3) we have $d(v)<\infty$ and $d(u) + w(u,v)<d(v)$. This means that node $v$ will remain in $\pq$ if it was already there, and will be moved from $R$ to $\pq$ if $d(u) + w(u,v) \leq P$. In both cases, the claim will remain true after the \dec\ operation when $d(v)$ is set to $d(u) + w(u,v)$. Lastly, the property also remains true when $P$ is inflated in the \restart, procedure, since all $v \in R$ with $d(v)\leq P$ are moved to $\pq$ and removed from $R$.
    
(P3) Both algorithms remove the same node $u$, update the tentative distance of $u$'s neighbours based on the same condition and update it to the same value. So if the claim holds before the iteration, it will also hold at the end of the iteration.
\end{proof}

\begin{restatable}{lemma}{lemmaReserveSet}
\label{lemma-reserveset}
The worst-case running time of all operations on the reserve set $R$ is $O(m + n)$.
\end{restatable}
\begin{proof}
If we let $w_{\max}$ be the maximum weight of an edge in the graph, then the maximum number of buckets we need is at most $\log_{\beta}(\frac{n w_{\max}}{\Pred}) = \log(\frac{n w_{\max}}{\Pred})/\log(\beta)$.
We conclude that the \restart\ procedure is called is at most $\log(\frac{n w_{\max}}{\Pred})/\log(\beta)$ many times.

We analyse the complexity of three operations on the reserve set. 
Firstly, we can insert a node into the reserve set in constant time by calculating the bucket number with the given tentative distance, after which we can insert it into the correct linked list. 
Secondly, a decrease priority operation is done by deleting the node first, after which it is inserted with the updated priority.
A decrease priority can therefore also be performed in constant time. 
Lastly, in the $i$'th \restart\ procedure, all the nodes of the $i$'th bucket are deleted from the reserve list and inserted into the standard priority queue. 
Furthermore, we note that the number of insertions into $R$ is bounded by the number of nodes, and the total number of decrease priorities is bounded by the number of edges in the graph.
This proves that all operations on the reserve set can be done in $O(m + n)$ time. 
\end{proof}

\section{Lower Bounds on Savings}\label{s:lower-bound}


In this section, we derive our lower bounds on the savings of \DijkstraPrediction.
%
%
We assume without loss of generality that all edge weights are normalised such that $w_e \in [0,1]$ for all $e \in E$.
Further, for the sake of the analysis, we assume that \DijkstraPrediction\ starts with a prediction $\Pred = \Dist + \mdelta$, where $\mdelta \ge 0$ is the additive error of the prediction. 
%
%
In particular, we assume that the algorithm starts with this prediction from the start (while it actually only becomes available after $i_{0}$ many iterations); but given that $i_0$ is small, this assumption is negligible.
Note that this clearly captures the case when $\Pred$ is an \emph{overestimation} of the actual distance $\Dist$. 
But our analysis also provides bounds on the number of priority queue operations when $\Pred$ is an underestimation of $\Dist$. 
To see this, 
let $\Pred_{under}$ be an underestimation. When the algorithm is started with $\Pred_{under}$, \restart\ will iteratively inflate the prediction until it exceeds $\Dist$. Let $\Pred_{over}$ be this final prediction. Then the number of priority queue operations of \DijkstraPrediction\ with \restart\ starting with $\Pred_{under}$ is dominated by the number of priority queue queue operations of \DijkstraPrediction\ with \restart\ starting with $\Pred_{over}$. (Note that in the run with $\Pred_{under}$ each node can only be inserted at a later stage into the priority queue than in the run with $\Pred_{over}$.)

\subsection{Worst-case Instances}

%
We first show that \DijkstraPrediction\ might not prune a single edge, even if the prediction is perfect (i.e., $\mdelta = 0$). To see this, suppose that an adversary can fix the entire input instance. Consider the instance depicted in Figure \ref{fig:lb_example1} ($t$ being the only target node).
A moments thought reveals that a necessary condition for an edge $(u_1, v_i) \in E$ to be pruned is that it belongs to the set 
$
    L := \sset{(u,v)\in E }{ d(u) \le \Dist \text{ and } d(u) + w(u,v) >\Dist}
$ (indicated in grey). 
However, none of these edges will be pruned (neither because of $B$ nor $P$)---and this holds even if the prediction is perfect (i.e., $\mdelta = 0$). 
The point here is that the distance $d(u_1) = \mdelta$ of the start node $u_1$ is small, and thus the tentative distance of $v_i$ cannot exceed $P$. 

Next, fix some threshold $\theta \in (D-1+\mdelta, \Dist]$ and define the set of \emph{relevant} edges as
\begin{align}
    L_\theta := \sset{(u,v)\in E }{ \theta \le d(u) \le \Dist \text{ and } d(u) + w(u,v) >\Dist}. \label{def:ltheta} 
\end{align}%
Suppose that an adversary can fix the instance as before, but now we enforce it to have many relevant edges. 
Even then none of these relevant edges might be pruned.
To see this, consider the illustration depicted in Figure \ref{fig:lb_example2}. Here, the nodes $u_i$ removed from the priority queue are sorted by increasing distances $d(u_i)$ (from left to right); only the $u_i$'s are shown for which $d(u_i)>\Dist-1+\mdelta$. Also, only the relevant edges in $L_\theta$ are shown (indicated in grey).
Without further restrictions, the adversary can still fix the weights of the relevant edges in $L_\theta$ as indicated such that none of these edges will be pruned (neither because of $B$ nor $P$). 
Note that this holds even for perfect predictions and $\theta$ being arbitrarily close to $\Dist$ (i.e., $\mdelta = 0$ and $\theta \rightarrow \Dist$).

The conclusion to draw from these examples is that our algorithm might not save on priority queue operations at all in the worst case. In essence, the crux here is that even though we have perfect information about the shortest path distance, this is not enough to speed-up the construction of the optimality certificate. Note that for both instances the distances of all nodes need to be determined correctly to obtain such a certificate. Given that \Dijkstra\ already uses the minimum number of priority queue operations to compute such a certificate, we cannot hope to improve on this. 

\begin{figure}[t]
\fbox{\begin{subfigure}{0.48\linewidth}
    \centering
    \includegraphics[page=1,width=\linewidth]{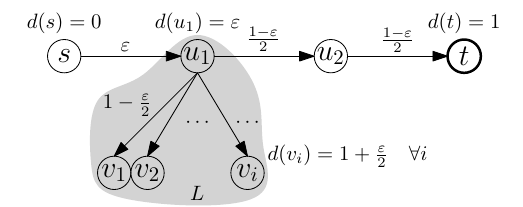}
    \caption{No pruning, even if predictions are perfect ($\mdelta=0$).}
    \label{fig:lb_example1}
\end{subfigure}}\hfill
\fbox{\begin{subfigure}{0.48\linewidth}
    \centering
    \includegraphics[page=2,width=\linewidth]{images/examples.pdf}
    \caption{No pruning, even if many relevant edges $(u_i, v_i) \in L_\theta$.} 
    \label{fig:lb_example2}
\end{subfigure}}
\label{fig:lb_examples}
\caption{Instances in which no savings can be expected from our algorithm}
\end{figure}

\subsection{Partial Random Instances}
Based on these examples,
it is clear that we need to further restrict the power of the adversary. We therefore introduce randomness in the instances to obtain a more fine-grained understanding of the savings achieved by our algorithm. Generally speaking, we will do this by enforcing randomness on some of the edges, while allowing the adversary to still control the rest of the input instance. 

The setup is as follows:
We suppose that the adversary can fix the set $Q = \set{u_{1}, \dots, u_{l}}$ of nodes that are removed from the priority queue, where $u_{i}$ is the node removed in iteration $i$, and the corresponding distances $d(u_1) \le \dots \le d(u_l)$ of these nodes.
Further, the adversary can fix all outgoing edges of the nodes in $Q$. Note that by doing so we implicitly allow that the adversary can fix the weights of certain edges to enforce this configuration (because these weights determine the order in which the nodes are removed from the priority queue). Crucially, however, we do not allow that the adversary can fix the weights of the relevant edges in $L_\theta$: 
the weight of each relevant edge in $L_\theta$ is random. 

Now, in this partially random setting, the weights for all edges $(u,v) \in L_\theta$ 
are independent and the distance labels $d(u) + w(u,v)$ are uniform on $[\Dist, d(u)+1]$ (\cite{bast2003heuristic}). This is because $(u,v)\in L_\theta$ implies $d(u) + w(u,v) \ge \Dist$ and $d(u)+w(u,v) \le d(u)+1$, but nothing else.

Given this partially random adversarial setting, we can lower bound the probability of pruning an edge in $L_\theta$. 
Like in the previous examples, the adversary still has enough power to setup the instance such that none of these edges are pruned based on the bound $B$. 
In contrast, each of these edges is pruned with positive probability due to the prediction $P$.
For the edges in $L_\theta$ we know that the necessary condition for pruning, $d(u) > \Dist+1-\mdelta$, holds. 
%

\begin{lemma}\label{lemma:succes}
Suppose $L_\theta$ is defined as above, with $\theta=\Dist + \gamma \mdelta -1$, for some $\gamma \in (1, \frac{1}{\mdelta}]$. Let $X_e$ be a random variable which equals 1 if edge $e\in L_\theta$ is pruned and zero otherwise. Then $
\pr{X_e = 1} \ge 1 - \frac{1}{\gamma}.$
\end{lemma}
\begin{proof}
Let $e=(u,v)\in L_\theta$ be a relevant edge and let $d(v)$ be the distance of $v$ at the end of the algorithm. Then $e$ is pruned whenever the tentative distance $\tent:=d(u) + w(u,v)$ exceeds the prediction $P$. As argued before, $e\in L_\theta$ implies that $d(u) + w(u,v)$ is uniformly distributed on $[\Dist, d(u) + 1]$. 
Furthermore, note that $1-(\Dist-d(u)) \ge \gamma \mdelta$ follows from $d(u) \ge \theta$ and the definition of $\theta$ and note that $\Dist +\mdelta$ is contained in the interval $[\Dist, d(u) + 1]$, since $\Dist + \mdelta \le \Dist + \gamma\mdelta \le d(u) + 1$. We can simply apply the cumulative distribution function for the uniform distribution: 
\[
\pr{X_e = 1}
= \pr{\tent \ge \Pred } 
= \pr{d(u) + w(u,v) \ge \Dist + \mdelta}
=  \frac{d(u) + 1 - (\Dist + \mdelta) }{d(u) + 1 - \Dist} 
\ge 1-\frac{1}{\gamma}.
\]
\end{proof}
We give some intuition for the lemma: 
The number of edges in $L_\theta$ is largest when $\gamma$ is close to 1, and the lemma shows there is a small positive probability for each of those to be pruned by the prediction.   
As $\gamma$ increases to $\frac{1}{\mdelta}$, the threshold $\theta$ approaches $\Dist$ and the size of $L_\theta$ decreases. So the number of relevant edges decreases, while the probability that each such edge is pruned increases.

Define $X$ as the total number of pruned edges in $L_\theta$, i.e., $X = \sum_{e \in L_\theta} X_e$. 
\begin{restatable}{theorem}{prunedEdgesWhp}
\label{thm:whp}
\label{thm:expected-savings}
Suppose $L_\theta$ is defined as above, with $\theta=\Dist + \gamma \mdelta -1$, for some $\gamma \in (1, \frac{1}{\mdelta}]$.
Then the expected number of pruned edges in $L_\theta$ is 
$
\pe{X} \ge (1 - \tfrac{1}{\gamma}) |L_\theta| .
$
Further, if $(1 - \frac{1}{\gamma}) |L_\theta|  \ge 8\ln n$, then 
$X \ge \frac12(1 - \frac{1}{\gamma}) |L_\theta|$ with high probability, i.e., 
$\pr{X \ge \frac{1}{2}(1 - \frac{1}{\gamma})|L_\theta|} \ge 1-\frac{1}{n}$.
\end{restatable}
\begin{proof}
By linearity of expectation, it follows from Lemma \ref{lemma:succes}:
\[
\pe{X} = \pel{\sum_{e \in L_\theta} X_e} = \sum_{e \in L_\theta} \pe{X_e} = \sum_{e \in L_\theta} \pr{X_e = 1}. 
\]

Note that $X = \sum_{e \in L_\theta} X_e$ is a sum of $|L_\theta|$ independent random variables. Let $\mu := \pe{X}$ be the expected value of $X$. The following (standard) Chernoff bound holds for every $\delta \in (0,1)$:
\[
\pr{X \le (1-\delta) \mu} \le e^{-\mu\delta^2/2}. 
\]
By choosing $\delta = \frac{1}{2}$, we obtain 
\begin{align*} 
\prl{X \le \frac12 \mu} \le e^{-\mu/8} \le \frac{1}{n},
\end{align*} 
where the second inequality holds because $\mu = \pe{X} \ge (1 - \frac{1}{\gamma}) |L_\theta| \ge 8 \ln n$ by our assumption. 

Using this, we conclude that 
\[
\prl{X \ge \frac{1}{2}\left(1 - \frac{1}{\gamma}\right)|L_\theta|}
\ge 
\prl{X \ge \frac12 \mu} 
\ge 
1-\frac{1}{n}. 
\]
\end{proof}

\subsection{Random model}

We consider random instances constructed according to the model introduced in Section \ref{sec:preliminaries}. For these instances, we show that our \DijkstraPrediction\ algorithm saves a significant number of queue operations compared to \DijkstraPruning. 

In the random model, it is not straightforward to compute the size of $L_\theta$. Therefore, we will consider a specific subset of $L_\theta$ for which we \emph{are} able to compute the size. More specifically, we only consider the edges $(u,v)$ from $L_\theta$ for which the \emph{final} distance $d(v)$ is larger than $\Dist$, i.e., edges from $L_\theta$ that lead to a node $v$ which is not in $Q$.\footnote{Note that $L_\theta$ may contain edges $(u,v)$ with tentative distance $d(u) + w(u,v) > D$, but whose final distance $d(v) \le D$. These are relevant edges having both endpoints in $Q$ that might be pruned. However, we do not account for these savings in our analysis here.}
Note that there could be multiple edges in $L_\theta$ that lead to such a node $v\not\in Q$. In that case, we only consider the (unique) edge in $L_\theta$ which has led to an insertion of $v$ into the priority queue in the standard \Dijkstra\ algorithm (disregarding edges that have led to a decrease priority operation). We use $L'_\theta$ to denote this subset of $L_\theta$:
\[
    L'_\theta := \sset{(u,v)\in L_\theta}{\text{$(u,v)$ leads to insertion of $v$ in \pq\ of \Dijkstra}}.
\]


In the \Dijkstra\ algorithm, all the end nodes of edges in $ L_\theta'$ are inserted in the priority queue, but they are never removed. 
\DijkstraPrediction\ can actually save a number of these insert operations by pruning the edges in $L_\theta'$.
We will lower bound these savings by computing an upper bound for the number of these nodes which are still inserted in our \DijkstraPrediction\ algorithm. 
Consequently, all the edges which lead to nodes which are not inserted in \DijkstraPrediction\ are pruned. 

First, we will prove the following Key Lemma which will help us to upper bound the probability of inserting an edge below. Below, we use $[k]$ to denote the set $\set{1,\ldots,k}$.


\begin{restatable}[Key Lemma]{lemma}{keyLemma}\label{Lemma:uniformvars}
Let $X_j, j =1,\ldots,k+1$, be $k+1$ uniform random variables, with $X_j$ uniform on $[a,b_j]$ and $b_1 \le \cdots \le b_{k+1}$. Let $P>0$ be a real number, which is contained in all intervals, i.e.,  $a<P<b_1$. Let $\event{1}$ be the event $\{X_{k+1} \le X_j, \forall j\in [k]\}$ and let $\event{2}$ be the event $\{X_{k+1} \le P\}$. Then:
$
\textstyle
\pr{\event{1} \wedge \event{2} } \le \frac{1}{k+1} ( 1- (1- \frac{P-a}{b_k-a})^{k+1}).
$
\end{restatable}
\begin{proof}
We will upper bound the probability by conditioning on values of $X_{k+1}$, using the law of total probability and applying the density function of $X_{k+1}$: $f_{X_{k+1}}(s) = \frac{1}{b_{k+1}-a}$.
\begin{align*}
    \pr{ \event{1} \wedge \event{2}  }  & = \int_a^{b_{k+1}} \pr{\event{1} \wedge \event{2}  \pcond X_{k+1} = x} f_{X_{k+1}}(x) dx\\
    &=  \frac{1}{b_{k+1}-a} \int_a^{b_{k+1}} \pr{\event{1} \wedge \event{2}  \pcond X_{k+1} = x}  dx
\end{align*}
Since $\pr{\event{1} \wedge \event{2}  \pcond X_{k+1} = x} = 0$ if $x > P$, we can write:
\begin{align*}
    \pr{\event{1} \wedge \event{2}  \pcond X_{k+1} = x} &= \pr{\event{1} \wedge \event{2}\pcond \{X_{k+1} = x\} \wedge \{x\le P \}}\\
    &=\pr{\{X_{k+1} \le X_j, \forall j \}\wedge \{X_{k+1} \le P\}\pcond \{X_{k+1} = x\} \wedge \{x\le P \}}\\
    &= \pr{x \le X_j, j=1,\ldots,k \pcond \{X_{k+1} = x\} \wedge \{x\le P \}}\\
    &= \pr{x \le X_j, j=1,\ldots,k \pcond x\le P }\\
    &= \prod_{j=1}^k \pr{x \le X_j \pcond x\le P }\\
    &= \prod_{j=1}^k \left( 1- \frac{x-a}{b_j-a} \right)\\
    &\le \left( 1- \frac{x-a}{b_k-a} \right)^k.
\end{align*}
The third equality holds because the conditioning already implies that $X_{k+1} \le P$. The fourth equality holds since the value of $X_j, j=1,\ldots,k$ is independent of the value of $X_{k+1}$. Thereafter we use that all the $X_j$'s are identically distributed, and we use the cumulative distribution function of the uniform distribution. 
In the last inequality we exploit that $b_j-a\le b_k-a$ for all $j$.  
We can use this, together with $\frac{b_k-a}{b_{k+1}-a} \le 1$, to upper bound the expectation of $\event{1}\wedge \event{2}$:
\begin{align*}
    \pr{ \event{1} \wedge \event{2}  }  & \le \frac{1}{b_{k+1}-a} \int_a^{P} \left( 1- \frac{x-a}{b_k-a} \right)^k  dx\\
      & = \frac{1}{b_{k+1}-a} \left[-\frac{b_k-a}{k+1} \left(1- \frac{x-a}{b_k-a}\right)^{k+1} \right]_a^P \\
      & = \frac{1}{b_{k+1}-a} \left[-\frac{b_k-a}{k+1} \left(1- \frac{P-a}{b_k-a}\right)^{k+1} + \frac{b_k-a}{k+1}\right] \\
      & = \frac{b_k-a}{b_{k+1}-a} \frac{1}{k+1}\left[1- \left(1- \frac{P-a}{b_k-a}\right)^{k+1}\right] \\
      & \le \frac{1}{k+1}\left[1- \left(1- \frac{P-a}{b_k-a}\right)^{k+1}\right]. 
\end{align*}
\end{proof}

We will continue to lower bound the expected number of end nodes of edges in $L_\theta'$ which are inserted in the \DijkstraPrediction\ algorithm, despite the prunings. We call this quantity $\inrp_\theta$.
We condition on the event $E_l$, 
which implies not only the adversarial setting introduced in the previous section, but also that the size of $L_\theta'$ equals $l$. 


\begin{restatable}{theorem}{INRPtheorem}\label{thm:INRP}
Suppose $L_\theta'$ is defined as above, with $\theta=\Dist + \gamma \mdelta -1$, for some $\gamma \in (1, \frac{1}{\mdelta}]$. Let $\inrp_\theta$ be the number of end nodes of edges in $L_\theta'$ which are inserted but never removed in the \DijkstraPrediction\ algorithm.  Under the conditioning of the event $E_l$, i.e., $|L_\theta'|=l$ and the adversarial setting, we have that:
$
\textstyle
\pe{\inrp_\theta\pcond E_l} \le \frac{1}{q} ( 1 + \ln(\frac{lq}{\gamma})).
$
\end{restatable}
\begin{proof}
Let $l$ be the size of $L_\theta'$ and let $e_1=(u_1,v_1), e_2=(u_2,v_2), \ldots, e_l=(u_l,v_l)$ be all the edges in $L_\theta'$. Note that there might be repetitions in the $u_i$'s, but all the $v_i$'s are distinct. For $i=1,\ldots,l$, define $X_i=d(u_i) + w(e_i)$. We observed in the previous section that for all edges $e_i$ in $L_\theta'$ it holds that $X_i$ is random uniform on $[\Dist, d(u_i)+1]$. 

In \DijkstraPruning\, $e_i$ leads to an insertion only if $X_i$ is smaller than $X_j$ for every free $v_j$, with $j<i$. Suppose that there are $k$ of these free $v_j$'s preceding $v_i$ in the endpoints of $L_\theta'$. In \DijkstraPrediction, an extra condition must be met, namely that $X_i$ does not exceed the prediction. To lower bound the expectation of $\inrp_\theta$, we partition over $k$, the number of free $v_j$ preceding $v_i$: 
$$
\pe{\inrp_\theta\pcond E_l} \le \sum_{1 \le i \le l} \sum_{0 \le k < i} \binom{i-1}{k} q^k (1-q)^{i-1-k} \pr{\{X_i \le X_j, \forall j\in [k]\} \wedge\{ X_i \le P\}}
$$
To be able to apply our Key Lemma (Lemma \ref{Lemma:uniformvars}), we need that the variables $X_j$ must be random uniform. We have already shown that they are random uniform on $[\Dist, d(u_j) +1]$, for which the upper bounds increase as $j$ increases. Moreover, we need that $P<b_1$, which holds since for all edges in $L_\theta'$ we have $d(u_i) + 1 > P$. Therefore, we can apply our Key Lemma to upper bound the probability in the sum above by 
\begin{align*}
 \textstyle \frac{1}{k+1}\left[1- \left(1- \frac{P-\Dist}{d(u_k)+1-\Dist}\right)^{k+1}\right] 
    \textstyle \le \frac{1}{k+1}\left[1- \left(1- \frac{1}{\gamma}\right)^{k+1}\right],
\end{align*}
which gives
\begin{align}
\pe{\inrp_\theta\pcond E_l} 
&  \le \sum_{1 \le i \le l} \sum_{0 \le k < i} \binom{i-1}{k} q^k (1-q)^{i-1-k} \frac{1}{k+1} \label{eq:line1} \\
&   - \sum_{1 \le i \le l} \sum_{0 \le k < i} \binom{i-1}{k} q^k (1-q)^{i-1-k} \frac{1}{k+1} \label{eq:line2} \left( 1-\frac{1}{\gamma}\right)^{k+1}.
\end{align}

We know from \cite{bast2003heuristic} that \eqref{eq:line1} is equal to 
$\sum_{1 \le i \le l} \frac{1}{iq} (1 - (1-q)^i)$.
We will use similar techniques to obtain such an expression for \eqref{eq:line2}. 
As in \cite{bast2003heuristic}, we use $\binom{i-1}{k} \frac{1}{k+1} = \frac{1}{i}\binom{i}{k+1}$, and then use the binomium of Newton to rewrite the sum: 
\begin{align*}
&\eqref{eq:line2}=\sum_{1 \le i \le l} \frac{1}{iq} \sum_{0 \le k < i} \binom{i}{k+1} \left(q\left(1-\frac{1}{\gamma}\right)\right)^{k+1} (1-q)^{i-(k+1)} \\
&=\sum_{1 \le i \le l} \frac{1}{iq} \left[ \sum_{0 \le k \le i} \binom{i}{k} q^k\left(1-\frac{1}{\gamma}\right)^{k} (1-q)^{i-k} - (1-q)^i\right]=\sum_{1 \le i \le l} \frac{1}{iq} \left[ \left(1-\frac{q}{\gamma}\right)^i - (1-q)^i\right].
\end{align*}

Combining these two bounds, we obtain 
\begin{align*}
\pe{\inrp_\theta\pcond E_l} 
& \le \sum_{1 \le i \le l} \frac{1}{iq} \left[1 - (1-q)^i -  \left(1-\frac{q}{\gamma}\right)^i + (1-q)^i\right]
= \sum_{1 \le i \le l} \frac{1}{iq} \left[1 -  \left(1-\frac{q}{\gamma}\right)^i \right].
\end{align*}
As in \cite{bast2003heuristic}, we split the sum at a yet to be determined index $i_0$. For $i<i_0$, we estimate $(1-(1-\frac{q}{\gamma})^i)\le \frac{iq}{\gamma}$, and for $i\ge i_0$, we use $(1-(1-\frac{q}{\gamma})^i)\le 1$. We obtain:
\begin{align*}
\pe{\inrp_\theta\pcond E_l} & \leq \frac{i_0}{\gamma} +  \frac{1}{q} \sum_{i_0 \le i \le l} \frac{1}{i} \approx \frac{i_0}{\gamma} + \frac{1}{q} \ln \left( \frac{l}{i_0} \right) \leq  \frac{1}{q} \left( 1 + \ln\left(\frac{lq}{\gamma}\right)\right),
\end{align*}
where the last inequality follows by noting that $\frac{i_0}{\gamma} + \frac{1}{q} \ln ( \frac{l}{i_0} )$ attains its minimum with respect to $i_0$ at $i_0=\frac{\gamma}{q}$.
\end{proof}

Now, suppose $\Dist$ lies in $[0,1-\mdelta)$ and let $\gamma$ be $\frac{1-\Dist}{\mdelta}$, which makes $\theta$ equal to 0. 
This means that all the edges which lead to nodes that are inserted but not removed by \Dijkstra\ are in the set $L_\theta'$. 
Said differently, the size of $|L_\theta'|$ is equal to the number of nodes that are inserted but never removed in the priority queue by the \Dijkstra\ algorithm. 
\cite{bast2003heuristic} estimate the expected value of this quantity, conditional on that many nodes are reachable from $s$ (“$R$ is large”). We summarise their findings in the following proposition. 

\begin{proposition}[\citet{bast2003heuristic}]\label{prop:Bastetal}
Consider an instance from the random model introduced in Section \ref{sec:preliminaries}. Let $R$ be the number of reachable nodes from $s$ in the random graph. Then if $c\ge 8$ and $f\ge 4 \ln(n)$, it holds that $R$ is large, i.e. $R \geq (1-\delta)\alpha n$, for some $\alpha$ such that $\alpha = 1-\exp(-c\alpha)$ and $\delta$ small like 0.01. Moreover, the expected number of nodes that are inserted but never removed in the priority queue by the \Dijkstra\ algorithm, given that $R$ is large, is approximately:
$$\pe{INRS \pcond R \text{ is large}} \approx \frac{c-1}{q}.$$
\end{proposition}

By exploiting this proposition, we can drop the dependency on the size of $L_\theta$ and we obtain the following theorem. 
\begin{restatable}{theorem}{finalSavings}
\label{thm:finalSavings}
Suppose $\Dist$ lies in $[0,1-\mdelta)$. Let $\inrp$ be the number of nodes that are inserted but never removed by \DijkstraPrediction. 
Then 
$
\textstyle\pe{\inrp\pcond  R \text{ is large}\,} \le  \frac{1}{q} \left( 1 + \ln(c-1) - \ln\left(\frac{1-\Dist}{\mdelta}\right)\right).$
\end{restatable}
\begin{proof}
We can upper bound the expected value of $\inrp_\theta$ given in Theorem~\ref{thm:INRP} as follows: Since $\ln(\frac{lq}{\gamma})$ is a convex function of $l$, we can replace $l$ simply by the expectation of $\inrs$ (using Proposition~\ref{prop:Bastetal}).
We obtain 
\begin{align*}
\textstyle\pe{\inrp\pcond  R \text{ is large}\,} 
&\textstyle\le \frac{1}{q} \left( 1 + \ln\left(\frac{(c-1)q\mdelta}{q(1-\Dist)}\right)\right)
= \frac{1}{q} \left( 1 + \ln(c-1) - \ln\left(\frac{1-\Dist}{\mdelta}\right)\right).
\end{align*}
\end{proof}

Let $\inrr$ denote the number of nodes that are inserted but never removed by \DijkstraPruning. It is shown in \cite{bast2003heuristic} that $\pe{\inrr\pcond  R \text{ is large}\,} \le \frac{1}{q} \left( 1 + \ln(c-1)\right).$
That is, compared to the \DijkstraPruning\ algorithm, our \DijkstraPrediction\ algorithm saves  $\frac{1}{q}\ln(\frac{1-\Dist}{\mdelta})$ insertions of such nodes. 
So even though \DijkstraPruning\ already saves a significant number of insertions, \DijkstraPrediction\ is able to further improve on this. 
Naturally, these savings grow whenever the prediction becomes more accurate and $\mdelta$ decreases. 

In our experiments, we consider random instances with $n=1000$, $c=8$ and $q=0.02$. For these instances, $\Dist$ is approximately 0.55. With a prediction $P=\Dist + \mdelta$ which overestimates $\Dist$ by at most $\mdelta = 0.1$ (which seems reasonable from the experiments), the expected number of $\inrp$ of \DijkstraPrediction\ is at most 63. 
In comparison, the expected number of $\inrs$ of \Dijkstra\ is 350; so our algorithm saves at least 287 of these insertions. 
The expected number of $\inrr$ of \DijkstraPruning\ is at most 137; our algorithms significantly improves upon this by exploiting the prediction.


\section{Prediction Methods}\label{s:prediction_methods}

The \prediction\ algorithm used in our algorithm \DijkstraPrediction\ can be obtained in numerous ways. Below, we explain how we obtain a prediction algorithm based on a machine learning approach. We elaborate on two different machine learning models and compare them to a benchmark prediction. Moreover, two alternative prediction methods based on breadth-first search (BFS) are given.
%
%

\subsection{ML-based Predictions}

In order to make a prediction after $\len$ iterations, we need to be able to describe the current optimisation run by means of some characteristic features. One of the challenges here is to come up with features that capture the essence of the current run such that they can be used by the machine learning model to make a good prediction of the shortest path distance. We do this by keeping track of a lower and upper bound on the shortest path distance in each iteration. More precisely, in iteration $i \le \len$, the distance $d(u)$ of the node $u$ extracted from the priority queue serves as the lower bound $d_i$ and the current value of the pruning bound $B$ is used as the upper bound $B_i$. The resulting sequence $X = ((d_1, B_1), (d_2, B_2), \dots, (d_{\len}, B_{\len}))$ of these lower and upper bounds for the first $\len$ iterations then constitutes what we call the \emph{trace} of the algorithm. 

A \emph{training sample} and \emph{target} for the machine learning algorithm then consists of the trace $X$ and the corresponding shortest path distance $\Dist$, respectively.
The set of samples for the machine learning models can be created by executing a run of \DijkstraPruning\ on each problem instance of the training set.
During this run, both the trace $X$ and the final shortest path distance $\Dist$ need to be stored. Before the traces are used to train the machine learning models, we normalise each feature by subtracting the mean and divide by the standard deviation.
To prevent blowing up the mean value of the upper bound feature, all bounds $B_i$ which are equal to the initial value of $B = \infty$ are set to $0$. 

We implemented and compared two standard machine learning models, namely a neural network model and a linear regression model. 
The neural network model that we use is a straightforward multilayer perceptron network consisting of two hidden layers, for which we optimise the number of nodes per layer by a \emph{$k$-fold cross validation} (see, e.g., \cite{Refaeilzadeh2009} for more details). 
To verify whether anything has been learned by these models at all, the results for these models are compared with a straightforward benchmark prediction. 
This benchmark prediction, which is independent of the instance, is computed by taking the average of the shortest path distance for each instance in the training set.

\subsection{BFS-based Predictions}

As an alternative to the machine learning models given in Section \ref{s:prediction_methods}, we implement two prediction methods that are based on breadth-first search (BFS). For each instance, we can simply run a BFS from the source node $s$ to determine a path $P^{\textit{BFS}}$ to any of the target nodes having the smallest number of edges. 
We also call $P^{\textit{BFS}}$ a \emph{BFS-path} and use $L^{\textit{BFS}}$ to denote its length (i.e., number of edges). 
Note that, equivalently, $L^{\textit{BFS}}$ is the shortest path distance to any of the target nodes if all edge weights are set to 1.
We use this BFS-path to derive two different predictions of the actual shortest path distance $\Dist$: 
(i) \textsc{bfs}: We define the prediction $\Pred$ as $L^{\textit{BFS}} \cdot \mu_w$, where $\mu_w$ is the expected edge weight.  
(ii) \textsc{$w$-bfs}: We define the prediction $\Pred$ as the sum of the \emph{actual} weights on $P^{\textit{BFS}}$, i.e., $\Pred = \sum_{e \in P^{\textit{BFS}}} w(e)$. 
Note that that the latter prediction might overestimate but never underestimate the actual distance $\Dist$. 


\section{Experimental Findings}
\label{sec:experiments}

In this section, we present our experimental findings. We first introduce our experimental setup and then discuss the results and insights we obtained from the experiments. 

\subsection{Experimental Setup}
We generated 100,000 instances of the SSMTSP problem using the random model described in Section~\ref{sec:preliminaries} with $n=1000$ nodes, edge probability $p = c/n$ with $c = 8$, and target probability $q = f/n$ with $f = 20$. The edge weights were chosen independently uniformly at random from $[0,1]$. 
Further, we fixed the length of the constructed traces to $\len = 10$. 
We only accepted an instance if \DijkstraPruning\ executed more than $\len$ iterations to ensure that \DijkstraLearning\ reaches the point where a prediction is made. 
This set of 100,000 instances was split into a \emph{training set} of 80,000 instances used for building the machine learning models, a \emph{validation set} of 10,000 instances used for parameter tuning and a \emph{test set} of 10,000 instances used for the final experiments. 

To get an idea of a few parameters related to the shortest path distance in the generated instances, we provide some statistical data for the validation and test set in Table~\ref{tbl:target_info}.
We computed the average number of edges on a shortest path, the average, minimum and maximum cost of a shortest path, and the average weight of an edge on a shortest path. Note that the first row refers to these parameters with respect to the actual random weights, while the second row refers to the case when all edge weights are set to $1$. 

\begin{table}[t]
\begin{center}
\begin{small}
\begin{sc}
\begin{tabular}{l|l|lll|l}
& edges & mean & min & max & $w$-mean \\
\hline
rand & 4.363   & 0.553            & 0.065           & 1.848           & 0.127  \\
unit   & 2.225   & 1.154            & 0.129           & 3.217           & 1.000      
\end{tabular}
\end{sc}
\end{small}
\end{center}
\caption{Statistics on different shortest path parameters for 20,000 instances (validation and test set).}
\label{tbl:target_info}
\end{table}

\subsection{Machine Learning Results}\label{ss:ml_results}

For each of the 80,000 graphs in the training set, we executed a run of \DijkstraPruning, during which we stored both the features $X$ and the final returned shortest path distance $y$. After running this, we had a training set of 80,000 samples, where each sample was of shape $(X,y)$. We used these 80,000 samples to built both a neural network and a linear regression model.

The number of nodes in the two hidden layers of the neural network was optimised by minimising the mean absolute error (MAE) in a $k$-fold cross validation with $k=4$ and a batch size of 256. We tested layer sizes $8, 16, 32, 64$ and $128$; 
the results for the smoothed validation Mean Absolute Error can be found in Figure \ref{fig:cross-fold}. We decided to use 16 nodes per layer and train for 47 epochs. We did not use the validation set of 10,000 instances in the $k$-fold cross validation, since it was used to tune parameters $\alpha$ and $\beta$. We also tested a linear regression and the benchmark prediction, but the neural network performed best in both the training and the test set. 

\begin{figure}[t]
\begin{center}
\centerline{\includegraphics[width=.4\columnwidth]{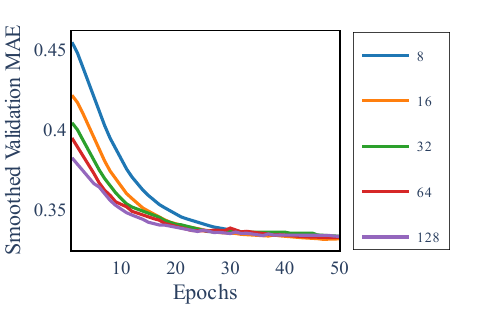}}
\vspace*{-2ex}
\caption{Normalised smoothed mean absolute error of the validation set for different hidden layer sizes.}
\label{fig:cross-fold}
\end{center}
\vskip -0.2in
\end{figure}

\begin{table}[t]
\begin{center}
\begin{small}
\begin{sc}
\begin{tabular}{l|lll}
Algorithm & NN & LR & AB \\\hline
Train MAE  & 0.0619     & 0.0883  & 0.1469        \\
Train MAPE & 0.1228     & 0.1844  & 0.3150        \\
Test MAE   & 0.0617     & 0.0880  & 0.1477        \\
Test MAPE  & 0.1217     & 0.1837  & 0.3160        \\
\end{tabular}
\end{sc}
\end{small}
\end{center}
\caption{Mean absolute error (MAE) and mean absolute percentage error (MAPE) for neural network (NN), linear regression (LR) and averaging benchmark (AB).}
\label{tbl:MAE}
\end{table}

The results for the performance of the neural network, linear regression and the benchmark are given in Table \ref{tbl:MAE}. 
Both machine learning models perform better than the benchmark prediction, in both the training and test set. 
Moreover, the neural network outperforms the linear regression model on both the training and the test set.

\subsection{Benchmark Algorithm \oracle}
In order to assess the performance of the different algorithms, we decided to use the following (idealised) benchmark algorithm to compare against: We run \DijkstraPruning\ with the pruning bound $B$ being \emph{initialised} with the actual shortest path distance $\Dist$. 
We refer to this algorithm as \oracle.

Note that this algorithm only inserts nodes into the priority queue which are necessary for finding the shortest path distance $\Dist$. Said differently, the algorithm spends the minimum possible amount of work to provide a certificate of optimality for the shortest path distance $\Dist$; no other algorithm could spend less work (as long as we insist that the shortest path distance is computed correctly always).


\subsection{Parameter Tuning}
There are two parameters in the \restart\ procedure which decide how to handle the machine learning prediction. The first one is $\alpha$, which specifies the amount by which the initial prediction is inflated. The second one is $\beta$, the amount by which the prediction is inflated. 
We investigate the impact of these parameters on the queue size; see  Figure~\ref{fig:qsize_and_distance} (bottom layer, both figures for same instance). On the left, we fix $\beta = 1.3$ and vary $\alpha$; on the right we fix $\alpha = 1.3$ and vary $\beta$. As is visible from these plots, a larger $\alpha$ means that the first call of \restart\ occurs later. Also, a larger $\beta$ leads to a larger number of nodes inserted during \restart.

We tested several configurations for $\alpha$ and $\beta$ on the instances in the validation set. 
Table~\ref{tbl:smart_alpha_beta_q} and Table~\ref{tbl:smart_alpha_beta_cumq} state the respective number of 
queue operations $Q$ and the cumulative queue size $C$ for various choices. As it turns out, for both these performance indicators it is best to choose $\alpha$ and $\beta$ small. 

\begin{table}[t]
\begin{center}
\begin{small}
\begin{sc}
\begin{tabular}{l|lllll}
$\alpha \backslash \beta $        & 1.05   & 1.10   & 1.20   & 1.50   & 2.00   \\ \hline
1.00 & \textbf{153.08} & 153.59 & 154.90 & 158.55 & 160.20 \\
1.05 & 154.37 & 154.80 & 155.83 & 158.08 & 158.84 \\
1.10 & 156.06 & 156.37 & 157.06 & 158.39 & 158.74 \\
1.20 & 160.38 & 160.50 & 160.82 & 161.20 & 161.30 \\
1.50 & 172.57 & 172.58 & 172.58 & 172.61 & 172.61 \\
2.00 & 182.27 & 182.27 & 182.27 & 182.27 & 182.27
\end{tabular}
\end{sc}
\end{small}
\end{center}
\caption{Average number of queue operations ($Q$) for \DijkstraPrediction\ for different values of $\alpha$ and $\beta$ on the validation set.}
\label{tbl:smart_alpha_beta_q}
\end{table}

\begin{table}[t]
\begin{center}
\begin{small}
\begin{sc}
\begin{tabular}{l|lllll}
$\alpha \backslash \beta $        & 1.05   & 1.10   & 1.20   & 1.50   & 2.00   \\ \hline
1.00 & \textbf{2446.0} & 2561.9 & 2740.4 & 3115.7 & 3274.6 \\
1.05 & 2573.5 & 2669.2 & 2815.5 & 3067.5 & 3150.1 \\
1.10 & 2732.6 & 2805.3 & 2914.2 & 3065.1 & 3104.4 \\
1.20 & 3112.6 & 3148.4 & 3197.7 & 3251.9 & 3266.3 \\
1.50 & 4104.3 & 4106.6 & 4109.5 & 4114.2 & 4114.2 \\
2.00 & 4809.9 & 4809.9 & 4809.9 & 4809.9 & 4809.9 
\end{tabular}
\end{sc}
\end{small}
\end{center}
\caption{Average cumulative queue size ($C$) for \DijkstraPrediction\ for different values of $\alpha$ and $\beta$ on the validation set.}
\label{tbl:smart_alpha_beta_cumq}
\end{table}

\subsection{Discussion of Results}
\begin{figure}
\centering
\begin{subfigure}{0.5\columnwidth}
  \centering
  \includegraphics[width=\columnwidth]{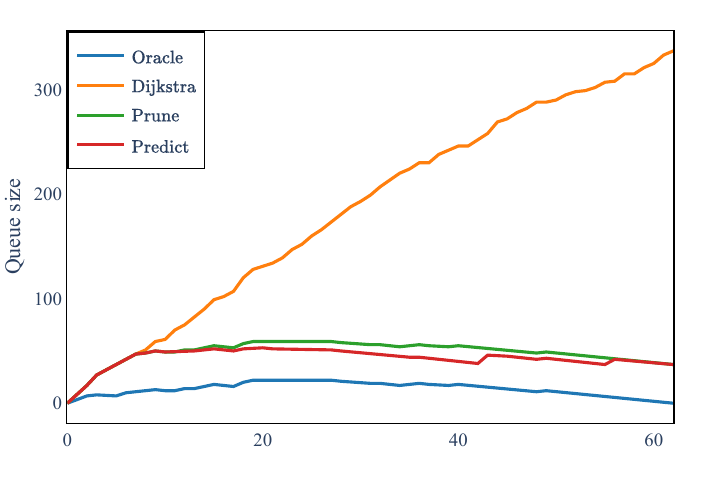}
\end{subfigure}%
\begin{subfigure}{0.5\columnwidth}
  \centering
  \includegraphics[width=\columnwidth]{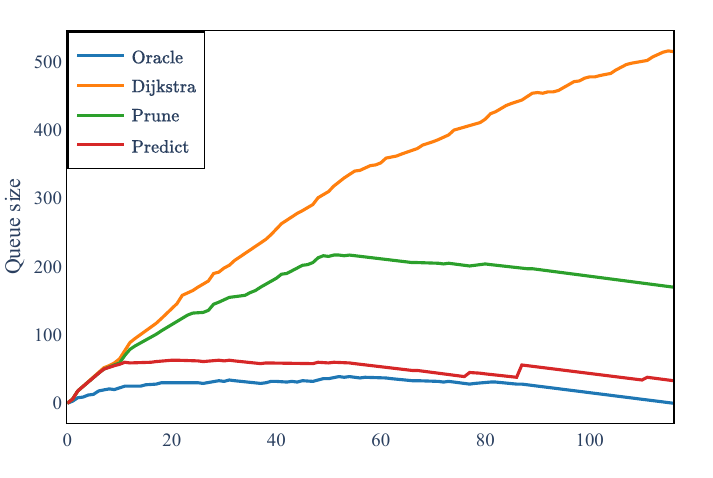}
\end{subfigure}
\begin{subfigure}{0.5\columnwidth}
  \centering
  \includegraphics[width=\columnwidth]{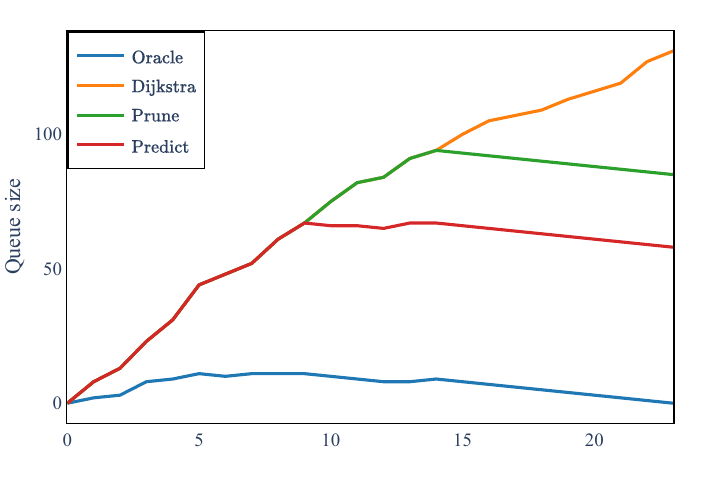}
\end{subfigure}%
\begin{subfigure}{0.5\columnwidth}
  \centering
  \includegraphics[width=\columnwidth]{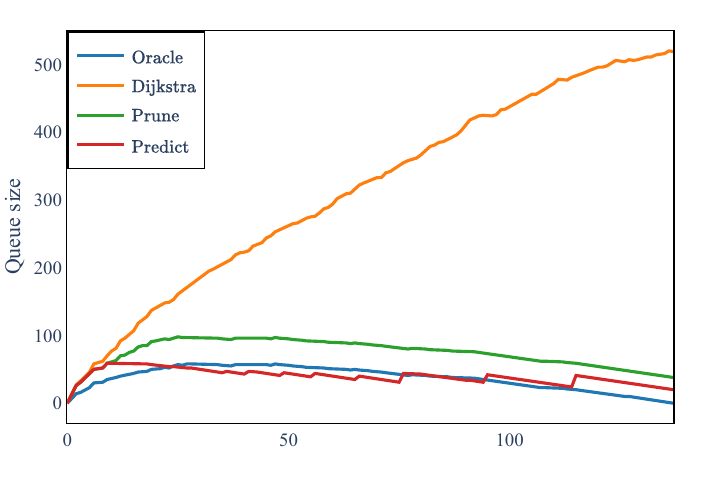}
\end{subfigure}
\begin{subfigure}{0.5\columnwidth}
  \centering
  \includegraphics[width=\columnwidth]{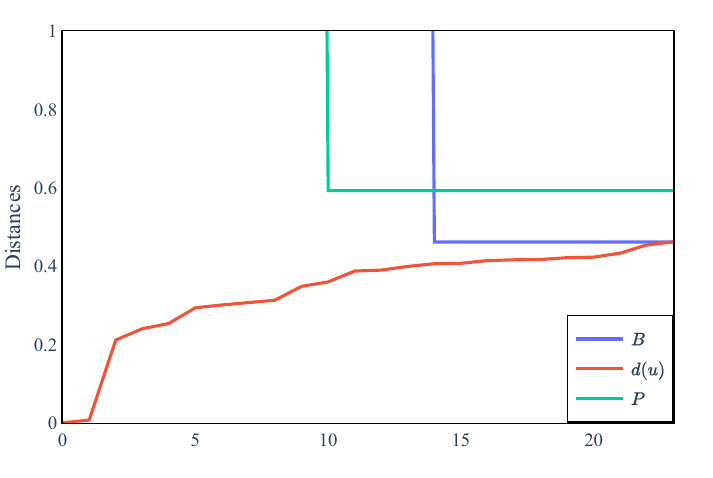}
\end{subfigure}%
\begin{subfigure}{0.5\columnwidth}
  \centering
  \includegraphics[width=\columnwidth]{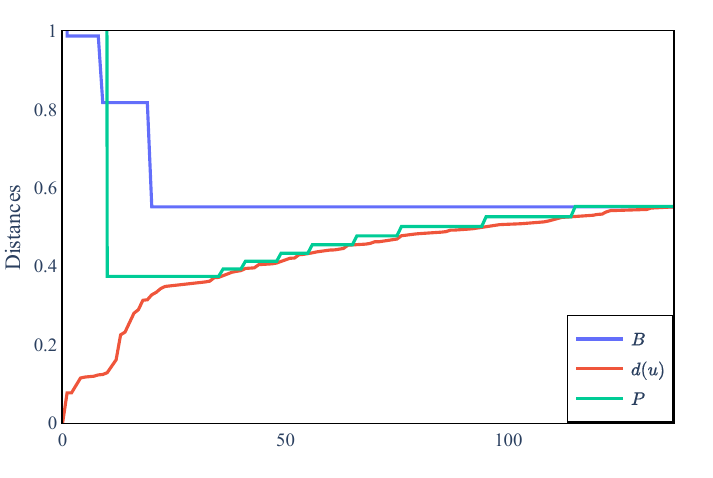}
\end{subfigure}%

\begin{subfigure}{0.49\columnwidth}
  \centering
  \includegraphics[width=\columnwidth]{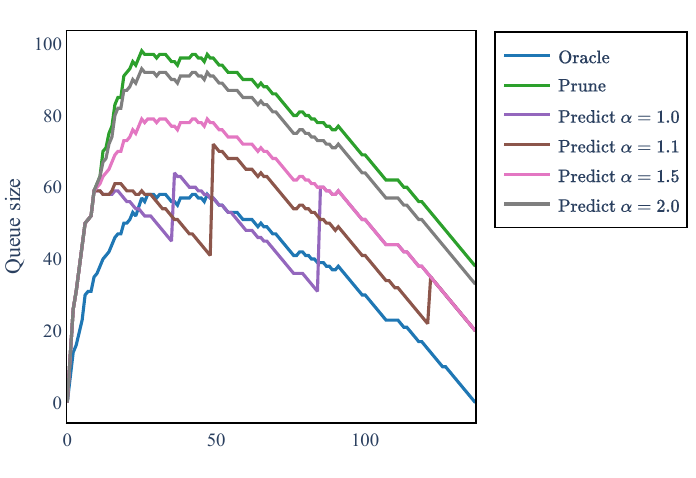}
\end{subfigure}%
\begin{subfigure}{0.49\columnwidth}
  \centering
  \includegraphics[width=\columnwidth]{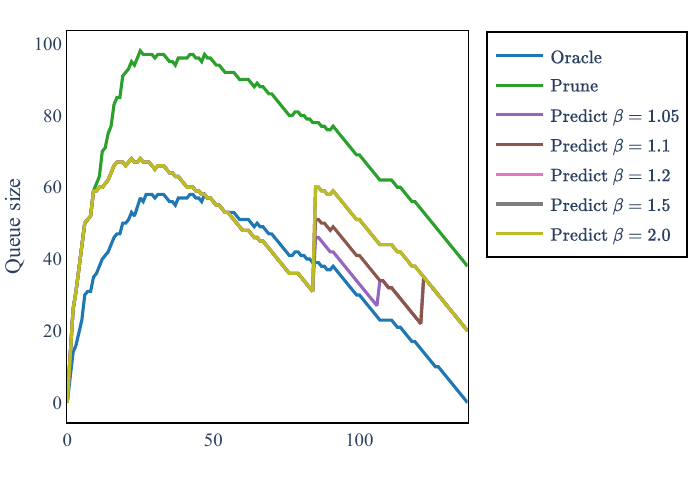}
\end{subfigure}
\caption{%
(In all plots the $x$-axis refers to the number of iterations.)
Top two layers (four distinctive instances): queue size of the algorithms. 
Third layer (same instances as in layer two): distance $d(u)$ together with $B$ and $\Pred$ for the \DijkstraPrediction. Bottom layer: queue size for fixed $\beta = 1.3$ and varying $\alpha$ (left) and fixed $\alpha = 1.3$ and varying $\beta$ (right) for \DijkstraPrediction. }
\label{fig:qsize_and_distance}
\end{figure}


We start by considering how the queue sizes differ for the different algorithms; see Figure \ref{fig:qsize_and_distance} top 2 rows. As to be expected, the queue size of \DijkstraLearning\ never exceeds the one of \DijkstraPruning\ and is larger or close to the one \oracle. The improvement of our \DijkstraPrediction\ with respect to \DijkstraPruning\ varies and depends on the instance. Next, we comment on the \restart\ procedure; see Figure \ref{fig:qsize_and_distance} third row. On the left, no restart was necessary since the prediction was not an overestimation of $D$. This plot also indicates the interplay of the prediction $\Pred$ and the pruning bound $B$; first the former and later the latter providing the smaller upper bound. On the right, we needed to do several \restart\ because the initial prediction turned out to be too small. \restart\ adds some nodes from the reserve list to the priority queue (queue size increases) and continues. 
After several inflations of the prediction with $\beta$, the prediction was sufficiently high to find the shortest path distance. 

If we zoom in to obtain a more fine-grained picture of the different queue operations executed by the algorithms, the results are as specified in Table~\ref{tbl:result1} (test set). 
The respective rows state the number of \del\ ($\textit{RM}$), \ins\ ($\textit{IS}$) and \dec\ ($\textit{DP}$) operations, the total number of queue operations ($Q$), the number of trials ($T$), the cumulative queue size $C$ (over all iterations), and the cumulative queue size relative to the \oracle\ $\bar{C}$. 
As to be expected, \oracle\ inserts and removes the minimum possible number of nodes only.  As can also be inferred from Invariant~\ref{invariant-all}, \Dijkstra, \DijkstraPruning\ and \DijkstraPrediction\ have the same number of \del\ operations.

Observe that the results show that our algorithm \DijkstraLearning\ outperforms all other algorithms, both in terms of the total number of queue operations and cumulative queue size. 
\DijkstraLearning\ outperforms \DijkstraPruning\ mostly on the number of insertions, as expected from the analysis in Section \ref{s:lower-bound}. 
In terms of cumulative queue size, our algorithm \DijkstraLearning\ even comes close to the benchmark \oracle, the average cumulative queue size being only 1.7 times larger then the one of \oracle
; \DijkstraPruning\ perform much worse, being off by a factor $3.6$. 

\begin{table}[t]
\begin{center}
\begin{small}
\begin{sc}
\begin{tabular}{l|r|rrrrr}
          & Oracle& Dijks & Prune  & Prediction   & bfs & $w$-bfs    \\ \hline
$\textit{RM}$         & 59.39    & 59.39   & 59.39     & 59.39   & 59.39      & 59.39  \\
$\textit{IS}$    & 59.39    & 335.50   & 122.91       & 91.73 & 117.66     & 118.86    \\
$\textit{DP}$       & 0.78  & 43.96    & 5.87         & 2.89    & 5.29       & 5.46    \\
$Q$       & 119.55   & 438.85   & 188.17     & 154.01  & 182.33     & 183.70   \\\hline
$T$       & 1.00     & 1.00    & 1.00      & 2.28     & 1.23       & 1.00    \\
$C$                 & 1456.16       & 13949.37 & 5245.96    & 2476.35 & 4825.09    & 4980.69              \\
$\bar{C}$           & 1.00& 9.58     & 3.60          & 1.70      & 3.31       & 3.42          
\end{tabular}
\end{sc}
\end{small}
\end{center}
\caption{Number of \del\ ($\textit{RM}$), \ins\ ($\textit{IS}$) and \dec\ ($\textit{DP}$) operations, the total number of queue operations ($Q$), the number of trials ($T$), the cumulative queue size $C$ (over all iterations), and the cumulative queue size relative to the \oracle\ $\bar{C}$ for \Dijkstra, \DijkstraPruning, \DijkstraLearning\ and the oracle algorithm. Averaged over all graphs in the test set.} 
\label{tbl:result1}
\end{table}


\subsection{Results for Different Graph Parameters}
In all results so far, we used a fixed set of parameters for the random graph model, namely an average degree $c$ of 8, and a target probability $q$ of 0.02. A natural question which might arise is how our algorithm performs on a less specific random graph structure. To answer this question, we built a new ML model, which is based on graphs with various random graph parameters, opposed to the single setting it was based on before. This new ML model showed us that, even though this new model is based on graphs with various input parameters, \DijkstraPrediction\ is still able to reduce the cumulative queue size. 

By taking $c$ from $\{2,5,8,16,32\}$ and $q$ from $\{0.02, 0.06, 0.18\}$, we created 15 pairs of random graph parameters. For each of these pairs, we constructed 80.000 graphs, which together formed a large training set of 1.2 million instances. We created a machine learning prediction model based on this training set, as explained before in Section \ref{ss:ml_results}. 

For each of the pairs of random graph parameters, we performed the \oracle, \Dijkstra, \DijkstraPruning\ and \DijkstraPrediction\ algorithm and compared the cumulative queue size of each algorithm to that of the \oracle. Table \ref{tbl:cq_results} shows the average relative cumulative queue size of a thousand instances. These results show us two things. Firstly and crucially, for each pair of random graph parameters, \DijkstraPrediction\ is able to reduce the cumulative queue size compared to \DijkstraPruning. This means that our algorithm does not lose its power to decrease the cumulative queue size, even when it is used on less specific random graph structures. 
For lower values of $q$, which means there are less target nodes in the instances, \DijkstraPrediction\ has a larger improvement over \DijkstraPruning\ than for larger values of $q$. 
A second and minor thing which these results show us is that the relative cumulative queue size can be lower than 1.0. E.g. when $c=2$ and $q=0.02$, \DijkstraPrediction\ has a lower cumulative queue size than \oracle. This seems unexpected, but can be explained by the average number of \restart\ routines done for those graph settings, which is 4.83. This fairly large number of \restart\ routines shows that the prediction was significantly too low and had to be increased several times. The low prediction caused the queue size to be low as well, which explains the small $\bar{C}$.

\begin{table*}[ht]
\begin{center}
\begin{small}
\begin{sc}
\begin{tabular}{l|rrr|rrr|rrr}
 & $q= 0.02$  &       &       & 0.06  &       &       & 0.18  &       &       \\ 
  & Dijkstra & Pruning & Predict & Dijkstra & Pruning & Predict & Dijkstra & Pruning & Predict \\ \hline
$c=2$ & 2.20  & 1.57  & 0.40  & 2.55  & 1.61  & 0.90  & 2.84  & 1.66  & 1.49  \\
5 & 6.10  & 2.62  & 1.01  & 7.71  & 2.92  & 2.19  & 8.86  & 2.95  & 2.82  \\
8 & 9.58  & 3.60  & 1.82  & 12.97 & 4.13  & 3.35  & 15.02 & 3.64  & 3.50  \\
16 & 16.04 & 4.76  & 2.84  & 24.78 & 5.26  & 4.51  & 31.20 & 5.07  & 4.92  \\
32 & 22.55 & 6.05  & 4.15  & 43.21 & 6.76  & 5.95  & 57.54 & 6.03  & 5.89 
\end{tabular}
\end{sc}
\end{small}
\end{center}
\caption{Cumulative queue sizes ($\bar{C}$) relative to the the cumulative queue size of \oracle\ for \Dijkstra, \DijkstraPruning\ and \DijkstraPrediction\ algorithm, for different random graph parameters $c$ and $q$.}
\label{tbl:cq_results}
\end{table*}

\subsection{Timing results}

Next to fine grained results counting operations on random graphs, we also tested the speed of our algorithm. We did this on certain graph instances which amplify the benefit of our algorithm. We call those instances the \emph{fortunate instances}, which we define in more detail below. Furthermore, we assumed that we had a perfect prediction in our pruning algorithm ($P=D$ and $\alpha=0$), and assumed this prediction was available from the start ($i_0=0$).

A fortunate graph instance is defined as follows (see Figure \ref{fig:fortunateinstance} for an example). Two input parameters must be given: the number of nodes, $n\in \mathbb{N}$, and the fraction of nodes which is on the shortest path, $r\in \mathbb{R}_+ $. From these parameters, the number of nodes on the shortest path, $x$, can be derived: $x:= \lfloor r\cdot n\rfloor$. We denote these nodes on the shortest path as $u_i$, for $i=0,\ldots, x-1$. Furthermore, we label $u_0$ as the start $s$ and label $u_{x-1}$ as the target $t$. The rest of the $n-x$ nodes in the graph are denoted by $v_j$, for $j=1, 2, \ldots, n-x$. Next to these $n$ nodes, there are $(x-1)(1+n-x)$ edges in the fortunate graph. There is an edge $(u_i, u_{i+1})$ for $i=0,1, \dots, x-2$, each with weight $\frac{1}{x-1}$. Together, these edges form the shortest path from $s$ to $t$, counting up to a shortest path length of exactly 1. Furthermore, there are edges $(u_i, v_j)$ for $i=0,1, \ldots, x-2$ and $j=1,2,\dots,n-x$. Each edge from $u_i$ to $v_j$ has length $2- \frac{2i}{x-1}$.

\begin{figure}[t]
\centering
\fbox{\includegraphics[width=0.8\linewidth]{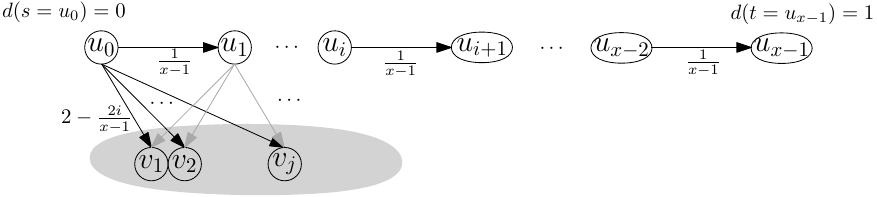}}
\caption{Fortunate instances used for timing analysis. Input parameters: the number of nodes $n\in \mathbb{N}$ and the fraction of nodes which is on the shortest path: $r\in \mathbb{R}_+ $. The number of nodes on the shortest path, $x\in \mathbb{N} $ can be derived from $n$ and $r$ as follows $x:= \lfloor r\cdot n\rfloor$.}
\label{fig:fortunateinstance}
\end{figure}

For $i=0, 1, \ldots, x-2$ and $j=1,2,\dots, n-x$, we define $P_{ij}$ to be the path going through the first $i+1$ nodes on the shortest path, after which it goes to $v_j$, i.e. $P_{ij} = (u_0, u_1, \ldots, u_i, v_j)$. Then the total length of $P_{ij}$ is as follows:
$$w(P_{ij}) = i \cdot \frac{1}{x-1} + 2 - \frac{2i}{x-1} = 2 - \frac{i}{x-1}.$$
This means that for a certain node $v_j$, the sequence of $w(P_{ij})$'s decreases as $i$ increases, but the path lengths never drop below the shortest path distance which equals 1. 
Since these $w(P_{ij})$'s strictly decrease, the tentative distance to $v_j$ will be updated in each iteration. In the pruning algorithm, this will result in a \dec\ operation in the priority queue. In the \textsc{Pruning} algorithm, this will result in a update of the tentative distance for a node in the reserve set.

\begin{table}[]
\begin{center}
\begin{small}
\begin{sc}
\begin{tabular}{l|ll|ll|ll|ll}
     & \multicolumn{2}{c}{Oracle} & \multicolumn{2}{c}{Dijkstra} & \multicolumn{2}{c}{Pruning} & \multicolumn{2}{c}{orpr} \\
    & time      & count        & time      & count        & time       & count       & time      & count        \\ \hline
RM  & 0 & 1250 & 0         & 1250         & 0         & 1250              & 0.01      & 1250         \\
IS  & 0 & 1250 & 0.01      & 5000         & 0         & 5000               & 0         & 1250         \\
DP  & 0 & 0 & 2.88      & 4680000      & 2.93      & 4676250            & 0         & 0            \\
RRM & -         & -            & -         & -            & -          & -           & 0         & 0            \\
RIS & -         & -            & -         & -            & -          & -           & 0.01      & 3750         \\
RDP & -         & -            & -         & -            & -          & -           & 2.22      & 4676250     
\end{tabular}
\end{sc}
\end{small}
\end{center}
\caption{Number of, and time span of, \del\ ($\textit{RM}$), \ins\ ($\textit{IS}$) and \dec\ ($\textit{DP}$) operations for the priority queue and deletions  (\textit{RRM}), insertions  (\textit{RIS}) and distance updates (\textit{RDP}) for the reserve set.
For a fortunate graph instance with $n = 5000$ and $r = 0.250$. For the \textsc{Prediction} algorithm, we assumed a perfect prediction which is known from the start: $P=D$, $\alpha = 1$ and $i_0=0$. A binomial heap structure was used for the priority queues in all four algorithms.}
\label{tbl:timeresult1}
\end{table}

We counted and timed the operations for the \textsc{Oracle}, \textsc{Dijkstra}, \textsc{Pruning} and \textsc{Prediction} algorithm for a fortunate graph instance with $n=5000$ and $r=0.25$, for which the results are shown in \ref{tbl:timeresult1}.
As expected due to the construction of the graph, there are many decrease priority operations, which take up most of the time. In the \textsc{Predict} algorithm, these nodes are never inserted into the priority queue, but remain in the reserve set. Therefore, there are many distance updates in the reserve set (exactly as many as DP operations in for \textsc{Pruning}). Crucially however, these distance updates take less time, which makes the \textsc{Predict} algorithm more efficient. 

We executed this test for multiple values of $r$ and for two different heap structures, binomial heap and fibonacci heaps. The run time can be seen in Table \ref{tbl:timeresult2}. Especially for the binomial heap and for large $r$, the benefit of our algorithm becomes clearly visible. For the binomial heap structure and the graph instance in which $r=0.35$, the \textsc{Predict} algorithm is 35\% faster than the fastest algorithm without prediction (\textsc{Dijkstra}).

\begin{table}[]
\begin{center}
\begin{small}
\begin{sc}
\begin{tabular}{ll|llll}
          &       &$ r = $0.05 & 0.1  & 0.25 & 0.35 \\ \hline
binomial  & Oracle: & 0        & 0.01 & 0    & 0.01 \\
          & Dijkstra: & 1.41     & 2.73 & 5.5  & 6.54 \\
          & Pruning: & 1.18     & 3.09 & 5.4  & 7.52 \\
          & Predict: & 1.3      & 2.28 & 4.2  & 4.82 \\\hline
fibonacci & Oracle: & 0.01     & 0    & 0.01 & 0.02 \\
          & Dijkstra: & 1.06     & 1.74 & 3.79 & 4.67 \\
          & Pruning: & 1.08     & 1.85 & 4.31 & 4.78 \\
          & Predict: & 1.03     & 2.04 & 4.49 & 5.08
\end{tabular}
\end{sc}
\end{small}
\end{center}
\caption{Time for different fortunate graph instances with $n=5000$ and varying fraction of nodes on the shortest path $r$. Two different heap structures where used for the priority queues in the different algorithms. }
\label{tbl:timeresult2}
\end{table}

\newpage
\bibliography{biblio.bib}

\begin{thebibliography}{55}
\providecommand{\natexlab}[1]{#1}
\providecommand{\url}[1]{\texttt{#1}}
\expandafter\ifx\csname urlstyle\endcsname\relax
  \providecommand{\doi}[1]{doi: #1}\else
  \providecommand{\doi}{doi: \begingroup \urlstyle{rm}\Url}\fi

\bibitem[Bagheri et~al.(2008)Bagheri, Akbarzadeh~Totonchi,
  et~al.]{bagheri2008finding}
A.~Bagheri, M.~R. Akbarzadeh~Totonchi, et~al.
\newblock Finding shortest path with learning algorithms.
\newblock \emph{International Journal of Artificial Intelligence}, 1, 2008.

\bibitem[Bast et~al.(2003)Bast, Mehlhorn, Sch{\"a}fer, and
  Tamaki]{bast2003heuristic}
H.~Bast, K.~Mehlhorn, G.~Sch{\"a}fer, and H.~Tamaki.
\newblock A heuristic for {D}ijkstra's algorithm with many targets and its use
  in weighted matching algorithms.
\newblock \emph{Algorithmica}, 36\penalty0 (1):\penalty0 75--88, 2003.

\bibitem[Bauer and Delling(2010)]{bauer2010sharc}
R.~Bauer and D.~Delling.
\newblock Sharc: Fast and robust unidirectional routing.
\newblock \emph{Journal of Experimental Algorithmics (JEA)}, 14:\penalty0 2--4,
  2010.

\bibitem[Bengio et~al.(2020)Bengio, Lodi, and Prouvost]{bengio2020machine}
Y.~Bengio, A.~Lodi, and A.~Prouvost.
\newblock Machine learning for combinatorial optimization: a methodological
  tour d’horizon.
\newblock \emph{European Journal of Operational Research}, 2020.

\bibitem[Chen et~al.(2022)Chen, Silwal, Vakilian, and Zhang]{pmlr-v162-chen22v}
J.~Chen, S.~Silwal, A.~Vakilian, and F.~Zhang.
\newblock Faster fundamental graph algorithms via learned predictions.
\newblock In K.~Chaudhuri, S.~Jegelka, L.~Song, C.~Szepesvari, G.~Niu, and
  S.~Sabato, editors, \emph{Proceedings of the 39th International Conference on
  Machine Learning}, volume 162 of \emph{Proceedings of Machine Learning
  Research}, pages 3583--3602. PMLR, 17--23 Jul 2022.
\newblock URL \url{https://proceedings.mlr.press/v162/chen22v.html}.

\bibitem[Chugh et~al.(2021)Chugh, Kumar, and Singh]{chugh2021survey}
G.~Chugh, S.~Kumar, and N.~Singh.
\newblock Survey on machine learning and deep learning applications in breast
  cancer diagnosis.
\newblock \emph{Cognitive Computation}, pages 1--20, 2021.

\bibitem[Cormen et~al.(2009)Cormen, Leiserson, Rivest, and
  Stein]{cormen2009introduction}
T.~H. Cormen, C.~E. Leiserson, R.~L. Rivest, and C.~Stein.
\newblock \emph{Introduction to algorithms}.
\newblock MIT press, 2009.

\bibitem[Dijkstra(1959)]{dijkstra1959note}
E.~W. Dijkstra.
\newblock A note on two problems in connexion with graphs.
\newblock \emph{Numerische mathematik}, 1\penalty0 (1):\penalty0 269--271,
  1959.

\bibitem[Dinitz et~al.(2021)Dinitz, Im, Lavastida, Moseley, and
  Vassilvitskii]{NEURIPS2021_5616060f}
M.~Dinitz, S.~Im, T.~Lavastida, B.~Moseley, and S.~Vassilvitskii.
\newblock Faster matchings via learned duals.
\newblock In M.~Ranzato, A.~Beygelzimer, Y.~Dauphin, P.~Liang, and J.~W.
  Vaughan, editors, \emph{Advances in Neural Information Processing Systems},
  volume~34, pages 10393--10406. Curran Associates, Inc., 2021.
\newblock URL
  \url{https://proceedings.neurips.cc/paper/2021/file/5616060fb8ae85d93f334e7267307664-Paper.pdf}.

\bibitem[Driscoll et~al.(1988)Driscoll, Gabow, Shrairman, and
  Tarjan]{driscoll1988relaxed}
J.~R. Driscoll, H.~N. Gabow, R.~Shrairman, and R.~E. Tarjan.
\newblock Relaxed heaps: An alternative to fibonacci heaps with applications to
  parallel computation.
\newblock \emph{Communications of the ACM}, 31\penalty0 (11):\penalty0
  1343--1354, 1988.

\bibitem[Eden et~al.(2022)Eden, Indyk, and Xu]{eden2022embeddings}
T.~Eden, P.~Indyk, and H.~Xu.
\newblock Embeddings and labeling schemes for a.
\newblock In \emph{13th Innovations in Theoretical Computer Science Conference
  (ITCS 2022)}. Schloss Dagstuhl-Leibniz-Zentrum f{\"u}r Informatik, 2022.

\bibitem[Elkin and Peleg(2004)]{elkin20041}
M.~Elkin and D.~Peleg.
\newblock (1+epsilon,$\beta$)-spanner constructions for general graphs.
\newblock \emph{SIAM Journal on Computing}, 33\penalty0 (3):\penalty0 608--631,
  2004.

\bibitem[Fakcharoenphol and Rao(2006)]{fakcharoenphol2006planar}
J.~Fakcharoenphol and S.~Rao.
\newblock Planar graphs, negative weight edges, shortest paths, and near linear
  time.
\newblock \emph{Journal of Computer and System Sciences}, 72\penalty0
  (5):\penalty0 868--889, 2006.

\bibitem[Fredman and Tarjan(1987)]{fredman1987fibonacci}
M.~L. Fredman and R.~E. Tarjan.
\newblock Fibonacci heaps and their uses in improved network optimization
  algorithms.
\newblock \emph{Journal of the ACM (JACM)}, 34\penalty0 (3):\penalty0 596--615,
  1987.

\bibitem[Fredman and Willard(1990{\natexlab{a}})]{fredman1990blasting}
M.~L. Fredman and D.~E. Willard.
\newblock Blasting through the information theoretic barrier with fusion trees.
\newblock In \emph{Proceedings of the twenty-second annual ACM symposium on
  Theory of Computing}, pages 1--7, 1990{\natexlab{a}}.

\bibitem[Fredman and Willard(1990{\natexlab{b}})]{fredman1990trans}
M.~L. Fredman and D.~E. Willard.
\newblock Trans-dichotomous algorithms for minimum spanning trees and shortest
  paths.
\newblock In \emph{Proceedings [1990] 31st Annual Symposium on Foundations of
  Computer Science}, pages 719--725. IEEE, 1990{\natexlab{b}}.

\bibitem[Fredman and Willard(1993)]{fredman1993surpassing}
M.~L. Fredman and D.~E. Willard.
\newblock Surpassing the information theoretic bound with fusion trees.
\newblock \emph{Journal of computer and system sciences}, 47\penalty0
  (3):\penalty0 424--436, 1993.

\bibitem[Gavoille et~al.(2004)Gavoille, Peleg, P{\'e}rennes, and
  Raz]{gavoille2004distance}
C.~Gavoille, D.~Peleg, S.~P{\'e}rennes, and R.~Raz.
\newblock Distance labeling in graphs.
\newblock \emph{Journal of Algorithms}, 53\penalty0 (1):\penalty0 85--112,
  2004.

\bibitem[Geisberger et~al.(2008)Geisberger, Sanders, Schultes, and
  Delling]{geisberger2008contraction}
R.~Geisberger, P.~Sanders, D.~Schultes, and D.~Delling.
\newblock Contraction hierarchies: Faster and simpler hierarchical routing in
  road networks.
\newblock In \emph{Experimental Algorithms: 7th International Workshop, WEA
  2008 Provincetown, MA, USA, May 30-June 1, 2008 Proceedings 7}, pages
  319--333. Springer, 2008.

\bibitem[Geisberger et~al.(2012)Geisberger, Sanders, Schultes, and
  Vetter]{geisberger2012exact}
R.~Geisberger, P.~Sanders, D.~Schultes, and C.~Vetter.
\newblock Exact routing in large road networks using contraction hierarchies.
\newblock \emph{Transportation Science}, 46\penalty0 (3):\penalty0 388--404,
  2012.

\bibitem[Gilbert(1959)]{Gilbert1959}
E.~N. Gilbert.
\newblock Random graphs.
\newblock \emph{The Annals of Mathematical Statistics}, 30\penalty0
  (4):\penalty0 1141--1144, 1959.

\bibitem[Gkatzelis et~al.(2022)Gkatzelis, Kollias, Sgouritsa, and
  Tan]{GKST2022}
V.~Gkatzelis, K.~Kollias, A.~Sgouritsa, and X.~Tan.
\newblock Improved price of anarchy via predictions.
\newblock In \emph{Proceedings of the 23rd ACM Conference on Economics and
  Computation}, EC '22, page 529?557, New York, NY, USA, 2022. Association for
  Computing Machinery.
\newblock ISBN 9781450391504.
\newblock \doi{10.1145/3490486.3538296}.
\newblock URL \url{https://doi.org/10.1145/3490486.3538296}.

\bibitem[Goldberg and Werneck(2005)]{goldberg2005computing}
A.~V. Goldberg and R.~F.~F. Werneck.
\newblock Computing point-to-point shortest paths from external memory.
\newblock In \emph{ALENEX/ANALCO}, pages 26--40, 2005.

\bibitem[Gutman(2004)]{gutman2004reach}
R.~J. Gutman.
\newblock Reach-based routing: A new approach to shortest path algorithms
  optimized for road networks.
\newblock \emph{ALENEX/ANALC}, 4:\penalty0 100--111, 2004.

\bibitem[Hagerup(2000)]{hagerup2000improved}
T.~Hagerup.
\newblock Improved shortest paths on the word ram.
\newblock In \emph{Automata, Languages and Programming: 27th International
  Colloquium, ICALP 2000 Geneva, Switzerland, July 9--15, 2000 Proceedings 27},
  pages 61--72. Springer, 2000.

\bibitem[Han(2001)]{han2001improved}
Y.~Han.
\newblock Improved fast integer sorting in linear space.
\newblock \emph{Information and Computation}, 170\penalty0 (1):\penalty0
  81--94, 2001.

\bibitem[Hilger et~al.(2009)Hilger, K{\"o}hler, M{\"o}hring, and
  Schilling]{hilger2009fast}
M.~Hilger, E.~K{\"o}hler, R.~H. M{\"o}hring, and H.~Schilling.
\newblock Fast point-to-point shortest path computations with arc-flags.
\newblock \emph{The Shortest Path Problem: Ninth DIMACS Implementation
  Challenge}, 74:\penalty0 41--72, 2009.

\bibitem[ich Lauther(2006)]{ich2006extremely}
U.~ich Lauther.
\newblock An extremely fast, exact algorithm for finding shor test paths in
  static networks with geographical background.
\newblock 2006.

\bibitem[K{\"o}hler et~al.(2005)K{\"o}hler, M{\"o}hring, and
  Schilling]{kohler2005acceleration}
E.~K{\"o}hler, R.~H. M{\"o}hring, and H.~Schilling.
\newblock Acceleration of shortest path and constrained shortest path
  computation.
\newblock In \emph{Experimental and Efficient Algorithms: 4th International
  Workshop, WEA 2005, Santorini Island, Greece, May 10-13, 2005. Proceedings
  4}, pages 126--138. Springer, 2005.

\bibitem[Kumar et~al.(2018)Kumar, Purohit, and Svitkina]{KPS2018}
R.~Kumar, M.~Purohit, and Z.~Svitkina.
\newblock Improving online algorithms via ml predictions.
\newblock In \emph{Proceedings of the 32nd International Conference on Neural
  Information Processing Systems}, NIPS'18, page 9684?9693, Red Hook, NY, USA,
  2018. Curran Associates Inc.

\bibitem[Lattanzi et~al.(2020)Lattanzi, Lavastida, Moseley, and
  Vassilvitskii]{onlinescheduling}
S.~Lattanzi, T.~Lavastida, B.~Moseley, and S.~Vassilvitskii.
\newblock Online scheduling via learned weights.
\newblock In \emph{Proceedings of the Fourteenth Annual ACM-SIAM Symposium on
  Discrete Algorithms}, pages 1859--1877. SIAM, 2020.

\bibitem[Lu and Weng(2007)]{lu2007survey}
D.~Lu and Q.~Weng.
\newblock A survey of image classification methods and techniques for improving
  classification performance.
\newblock \emph{International journal of Remote sensing}, 28\penalty0
  (5):\penalty0 823--870, 2007.

\bibitem[Lykouris and Vassilvitskii(2021)]{LV2021}
T.~Lykouris and S.~Vassilvitskii.
\newblock Competitive caching with machine learned advice.
\newblock \emph{J. ACM}, 68\penalty0 (4), jul 2021.
\newblock ISSN 0004-5411.
\newblock \doi{10.1145/3447579}.
\newblock URL \url{https://doi.org/10.1145/3447579}.

\bibitem[Madkour et~al.(2017)Madkour, Aref, Rehman, Rahman, and
  Basalamah]{madkour2017survey}
A.~Madkour, W.~G. Aref, F.~U. Rehman, M.~A. Rahman, and S.~Basalamah.
\newblock A survey of shortest-path algorithms.
\newblock \emph{arXiv preprint arXiv:1705.02044}, 2017.

\bibitem[Maue et~al.(2010)Maue, Sanders, and Matijevic]{maue2010goal}
J.~Maue, P.~Sanders, and D.~Matijevic.
\newblock Goal-directed shortest-path queries using precomputed cluster
  distances.
\newblock \emph{Journal of Experimental Algorithmics (JEA)}, 14:\penalty0 3--2,
  2010.

\bibitem[Medina and Vassilvitskii(2017)]{MV2017}
A.~M. Medina and S.~Vassilvitskii.
\newblock Revenue optimization with approximate bid predictions.
\newblock In \emph{Proceedings of the 31st International Conference on Neural
  Information Processing Systems}, NIPS'17, pages 1856--1864, Red Hook, NY,
  USA, 2017. Curran Associates Inc.
\newblock ISBN 9781510860964.

\bibitem[Mitzenmacher and Vassilvitskii(2020)]{algorithmspredictions}
M.~Mitzenmacher and S.~Vassilvitskii.
\newblock Algorithms with predictions.
\newblock \emph{CoRR}, abs/2006.09123, 2020.
\newblock URL \url{https://arxiv.org/abs/2006.09123}.

\bibitem[M{\"o}hring et~al.(2007)M{\"o}hring, Schilling, Sch{\"u}tz, Wagner,
  and Willhalm]{mohring2007partitioning}
R.~H. M{\"o}hring, H.~Schilling, B.~Sch{\"u}tz, D.~Wagner, and T.~Willhalm.
\newblock Partitioning graphs to speedup dijkstra's algorithm.
\newblock \emph{Journal of Experimental Algorithmics (JEA)}, 11:\penalty0 2--8,
  2007.

\bibitem[Otter et~al.(2020)Otter, Medina, and Kalita]{otter2020survey}
D.~W. Otter, J.~R. Medina, and J.~K. Kalita.
\newblock A survey of the usages of deep learning for natural language
  processing.
\newblock \emph{IEEE Transactions on Neural Networks and Learning Systems},
  32\penalty0 (2):\penalty0 604--624, 2020.

\bibitem[Papadimitriou and Steiglitz(1998)]{PS1998}
C.~H. Papadimitriou and K.~Steiglitz.
\newblock \emph{Combinatorial Optimization : Algorithms and Complexity}.
\newblock {Dover Publications}, July 1998.
\newblock ISBN 0486402584.

\bibitem[Qayyum et~al.(2020)Qayyum, Qadir, Bilal, and
  Al-Fuqaha]{qayyum2020secure}
A.~Qayyum, J.~Qadir, M.~Bilal, and A.~Al-Fuqaha.
\newblock Secure and robust machine learning for healthcare: A survey.
\newblock \emph{IEEE Reviews in Biomedical Engineering}, 14:\penalty0 156--180,
  2020.

\bibitem[Refaeilzadeh et~al.(2009)Refaeilzadeh, Tang, and
  Liu]{Refaeilzadeh2009}
P.~Refaeilzadeh, L.~Tang, and H.~Liu.
\newblock \emph{Cross-Validation}, pages 532--538.
\newblock Springer US, Boston, MA, 2009.
\newblock ISBN 978-0-387-39940-9.
\newblock URL \url{https://doi.org/10.1007/978-0-387-39940-9_565}.

\bibitem[Rizi et~al.(2018)Rizi, Schloetterer, and Granitzer]{Rizi_2018}
F.~S. Rizi, J.~Schloetterer, and M.~Granitzer.
\newblock Shortest path distance approximation using deep learning techniques.
\newblock \emph{2018 IEEE/ACM International Conference on Advances in Social
  Networks Analysis and Mining (ASONAM)}, Aug 2018.
\newblock \doi{10.1109/asonam.2018.8508763}.
\newblock URL \url{http://dx.doi.org/10.1109/ASONAM.2018.8508763}.

\bibitem[Sanders and Schultes(2005)]{sanders2005highway}
P.~Sanders and D.~Schultes.
\newblock Highway hierarchies hasten exact shortest path queries.
\newblock In \emph{ESA}, volume 3669, pages 568--579. Springer, 2005.

\bibitem[Sanders and Schultes(2006)]{sanders2006engineering}
P.~Sanders and D.~Schultes.
\newblock Engineering highway hierarchies.
\newblock In \emph{ESA}, volume~6, pages 804--816. Springer, 2006.

\bibitem[Schulz et~al.(2000)Schulz, Wagner, and Weihe]{schulz2000dijkstra}
F.~Schulz, D.~Wagner, and K.~Weihe.
\newblock Dijkstra's algorithm on-line: An empirical case study from public
  railroad transport.
\newblock \emph{Journal of Experimental Algorithmics (JEA)}, 5:\penalty0
  12--es, 2000.

\bibitem[Silver et~al.(2017)Silver, Schrittwieser, Simonyan, Antonoglou, Huang,
  Guez, Hubert, Baker, Lai, Bolton, et~al.]{silver2017mastering}
D.~Silver, J.~Schrittwieser, K.~Simonyan, I.~Antonoglou, A.~Huang, A.~Guez,
  T.~Hubert, L.~Baker, M.~Lai, A.~Bolton, et~al.
\newblock Mastering the game of go without human knowledge.
\newblock \emph{nature}, 550\penalty0 (7676):\penalty0 354--359, 2017.

\bibitem[Silver et~al.(2018)Silver, Hubert, Schrittwieser, Antonoglou, Lai,
  Guez, Lanctot, Sifre, Kumaran, Graepel, et~al.]{silver2018general}
D.~Silver, T.~Hubert, J.~Schrittwieser, I.~Antonoglou, M.~Lai, A.~Guez,
  M.~Lanctot, L.~Sifre, D.~Kumaran, T.~Graepel, et~al.
\newblock A general reinforcement learning algorithm that masters chess, shogi,
  and go through self-play.
\newblock \emph{Science}, 362\penalty0 (6419):\penalty0 1140--1144, 2018.

\bibitem[Sommer(2014)]{sommer2014shortest}
C.~Sommer.
\newblock Shortest-path queries in static networks.
\newblock \emph{ACM Computing Surveys (CSUR)}, 46\penalty0 (4):\penalty0 1--31,
  2014.

\bibitem[Thorup(1999)]{thorup1999undirected}
M.~Thorup.
\newblock Undirected single-source shortest paths with positive integer weights
  in linear time.
\newblock \emph{Journal of the ACM (JACM)}, 46\penalty0 (3):\penalty0 362--394,
  1999.

\bibitem[Thorup and Zwick(2005)]{thorup2005approximate}
M.~Thorup and U.~Zwick.
\newblock Approximate distance oracles.
\newblock \emph{Journal of the ACM (JACM)}, 52\penalty0 (1):\penalty0 1--24,
  2005.

\bibitem[van Emde~Boas(1975)]{van1975preserving}
P.~van Emde~Boas.
\newblock Preserving order in a forest in less than logarithmic time.
\newblock In \emph{16th Annual Symposium on Foundations of Computer Science
  (sfcs 1975)}, pages 75--84. IEEE, 1975.

\bibitem[Xu and Lu(2022)]{XuLu2022}
C.~Xu and P.~Lu.
\newblock Mechanism design with predictions.
\newblock In L.~D. Raedt, editor, \emph{Proceedings of the Thirty-First
  International Joint Conference on Artificial Intelligence, {IJCAI-22}}, pages
  571--577. International Joint Conferences on Artificial Intelligence
  Organization, 7 2022.
\newblock \doi{10.24963/ijcai.2022/81}.
\newblock URL \url{https://doi.org/10.24963/ijcai.2022/81}.
\newblock Main Track.

\bibitem[Zhao et~al.(2010)Zhao, Sala, Wilson, Zheng, and Zhao]{zhao2010orion}
X.~Zhao, A.~Sala, C.~Wilson, H.~Zheng, and B.~Y. Zhao.
\newblock Orion: shortest path estimation for large social graphs.
\newblock \emph{networks}, 1:\penalty0 5, 2010.

\bibitem[Zhao et~al.(2011)Zhao, Sala, Zheng, and Zhao]{zhao2011efficient}
X.~Zhao, A.~Sala, H.~Zheng, and B.~Y. Zhao.
\newblock Efficient shortest paths on massive social graphs.
\newblock In \emph{7th International Conference on Collaborative Computing:
  Networking, Applications and Worksharing (CollaborateCom)}, pages 77--86.
  IEEE, 2011.

\end{thebibliography}

\newpage
\appendix

\end{document}